\pdfoutput=1

\documentclass[final]{IEEEtran}

\IEEEoverridecommandlockouts                              
\usepackage{xcolor}

\usepackage{amsmath}
\usepackage{amsfonts}
\usepackage{amssymb}
\usepackage{bigints}
\usepackage{graphicx}
\usepackage{subcaption}
\usepackage{booktabs}
\usepackage{amsthm}
\usepackage{framed}

\usepackage[colorlinks=true]{hyperref}
\usepackage[bottom]{footmisc}

\newtheorem{theorem}{Theorem}[section]

\newtheorem{definition}{Definition}[section]

\newtheorem{remark}{Remark}[section]

\newtheorem{example}{Example}[section]

\newtheorem{problem}{Problem}[section]

\newcommand{\veps}{\varepsilon}

\newcommand{\sr}{\stackrel}

\newcommand{\rar}{\rightarrow}

\newcommand{\tri}{\sr{\triangle}{=}}

\newcommand{\be}{\begin{equation}}
\newcommand{\ee}{\end{equation}}
\newcommand{\bea}{\begin{eqnarray}}
\newcommand{\eea}{\end{eqnarray}}
\newcommand{\bes}{\begin{eqnarray*}}
\newcommand{\ees}{\end{eqnarray*}}

\newcommand{\bi}{\begin{itemize}}
\newcommand{\ei}{\end{itemize}}
\newcommand{\ben}{\begin{enumerate}}
\newcommand{\een}{\end{enumerate}}

\newcommand{\bp}{\begin{problem}}
\newcommand{\ep}{\end{problem}}
\newcommand{\hso}{\hspace{.1in}}
\newcommand{\hst}{\hspace{.2in}}

\newcommand{\noi}{\noindent}

\newcommand{\bc}{\begin{center}}
\newcommand{\ec}{\end{center}}

\newcommand{\quotes}[1]{``#1''}
\title{\LARGE \bf
General  Decentralized Stochastic  Optimal Control via Change  of Measure:   Applications to the Witsenhausen Counterexample}

\author{Bhagyashri Telsang$^1$, Seddik Djouadi$^1$, Charalambos D. Charalambous$^2$  
\thanks{$^{1}$Bhagyashri Telsang and Seddik Djouadi are  with the Faculty of Electrical Engineering and Computer Science, University of Tennessee, Knoxville, TN, 37996, USA 
{\tt\small \{btelsang, mdjouadi\}@utk.edu}}%
\thanks{$^{2}$Charalambos D. Charalambous is  with the Faculty of Electrical and Computer Engineering, University of Cyprus, Nicosia 1678, Cyprus
{\tt\small chadcha@ucy.ac.cy}}%
}

\begin{document}

\maketitle
\thispagestyle{empty}
\pagestyle{empty}

\begin{abstract}
In this paper we present  global and person-by-person (PbP) optimality  conditions  for general decentralized stochastic dynamic optimal control problems, using a  discrete-time version of Girsanov's change of measure. The PbP   optimality conditions are applied to the  Witsenhausen  counterexample to show that the two strategies satisfy two coupled nonlinear integral equations. Further, we prove  a fixed point theorem  in a function space, establishing existence and uniqueness of solutions to the integral equations.  We also provide  numerical solutions of the two  integral equations  using the Gauss Hermite Quadrature scheme, and  include a detail comparison to other numerical methods of the literature.  The numerical solutions confirm Witsehausen's observation  that, for certain choices of parameters, linear or affine strategies are optimal, while for other choices of parameters nonlinear strategies outperformed affine strategies.


\end{abstract}

\section{INTRODUCTION}
Witsenhausen in the 1971 paper \cite{witsenhausen1971}   introduced  a general  mathematical model for  decentralized stochastic dynamic optimal control problems operating over a finite discrete-time horizon $T_+^n\tri \{1,2,\ldots, n\}$, which  is used  to this date to model many features of   communication and  queuing networks,   networked control  systems applications etc.  The model consists of   multiple  observation posts collecting information at each time step, specified by $M$ measurements,   $\big\{y_t^m\big|t \in T_+^n\big\}, m=1, \ldots, M$,  multiple controls applied by $K$  stations, $\big\{u_t^k\big|t \in T_+^n\big\}, k=1, \ldots, K$, a Markov controlled state process, $\big\{x_t\big|t \in T_+^n\big\}$, and a payoff to be optimized by the strategies of the  controls.  At each time $t \in T_+^n$  control  actions are generated by strategies $\gamma^k(\cdot)$, i.e.,  $u_t^k=\gamma_t^k(I_t^k)$, where the information pattern $I_t^k$ is a  causal subset of all observations and all control actions,  for $k=1, \ldots, K$. 
 The derivation of optimality conditions for  Witsenhausen's  \cite{witsenhausen1971} general discrete-time model remains to this date open and challenging. 

The hardness of these optimization problems 
 is attributed to the {\it information pattern} of the  controls. Unlike 
classical  stochastic optimal control problems,  i.e.,   \cite{elliott-aggoun-moore1995,bensoussan1981,bensoussan1992a,elliott-yang1991,charalambous-elliott1997}, for decentralized stochastic dynamic optimal control problems,
\\
{\it i) the strategies do not have access to  the same causal  information pattern at each time instant,    and}\\
{\it ii) the  strategies may not have {\it perfect recall of the information pattern or structure}, i.e.,  any  information pattern  which is  accessible by any of the strategies  at any time $t$ may not be accessible  at all future  times
$\tau \geq t$. }


A simple and revealing example is 
Witsenhausen's  1968  \cite{WitsenhausenOriginal}, two-stage ``counterexample of stochastic optimal control",  shown in Fig.~\ref{fig:witsenhausenblkdia}. 
At the first stage the strategy of the control generates actions $u_1=\gamma_1(y_0)$ by observing the initial state   $y_0=x_0$,  while at the second stage the strategy  of the control generates actions  $u_2=\gamma_2(y_1)$ by observing $y_1$  but not $y_0$. Since strategy $\gamma_2$ does not observe both $(y_0, y_1)$, 
classical  stochastic optimal control methods, such as, dynamic programming,  do not  apply to the counterexample. Furthermore, since $y_1$ depends on strategy $\gamma_1$ via $x_1$, standard calculus of variations methods are not easily  applicable. 

Because of its significance, since it was introduced many   researchers have given a great deal of attention to the counterexample and variations of it
\cite{Papadimitriou-Tsitsiklis:1985,hierarchialLee,neuralnetworksolution,iterativecodingWit,BasarVariations,McEany1BadPaper,WuVerduTransport,McEany2,Submarmanian_newresults}. 
\noindent Past  literature  is mostly   focussed on   numerical searches of  the optimal payoff \cite{hierarchialLee,
neuralnetworksolution,iterativecodingWit,BasarVariations,McEany1BadPaper,WuVerduTransport,McEany2,
Submarmanian_newresults},  often  making  use of  properties derived by Witsenhausen  \cite[Theorem~1, Theorem~2,  Lemma~9, Lemma~11,  etc.]{WitsenhausenOriginal},  to reduce the computation burden. Another  extensively used property  is  
 \cite[Lemma~3.3.(c)]{WitsenhausenOriginal}, which states:  for fixed ${\gamma}_1(\cdot)$ the optimal strategy  $\gamma_2^o(\cdot)$  is the conditional mean (see Fig.~\ref{fig:witsenhausenblkdia}),  $u_2^o=\gamma_2^o(y_1)= {\bf  E}^{\gamma_1, \gamma_2^o} \big\{x_0+ {\gamma}_1(x_0)  \big| y_1 \big\}$. 

\begin{figure}
    \centering
    \includegraphics[width=\columnwidth]{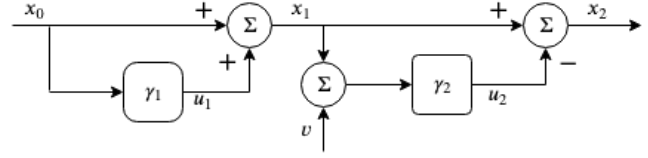}
    \caption{Witsenhausen's counterexample: $x_0: \Omega \rar {\mathbb R}$ is the initial state random variable (RV), $v: \Omega \rar {\mathbb R}$ is a  noise RV, $x_0$ and $v$ are independent, and   $x_1=x_0+\gamma_1(y_0)$, $x_2=x_1-\gamma_2(y_1)$, $y_0=x_0, y_1=x_1+v$. The objective is to minimize over $(\gamma_1, \gamma_2)$ the payoff,
$J(\gamma_1,\gamma_2) \tri   {\bf  E}^{\gamma_1, \gamma_2} \Big\{ k^2 (\gamma_1(y_0))^2 + (x_2)^2 \Big\}, k^2>0$. The counterexample is called linear-quadratic-Gaussian (LQG) if $x_0\in G(0, \sigma_x^2)$ and $v\in G(0, \sigma^2)$, where $G(\alpha, \beta^2)$ denotes a  Gaussian distribution with mean $\alpha$ and variance $\beta^2>0$.    }
    \label{fig:witsenhausenblkdia}
\end{figure}

\indent In  the  early 1970's,  it is recognized   that concepts from {\it static team theory},    developed by  Marschak and Radner  \cite{marschak1955,radner1962,marschak-radner1972} (see also \cite{krainak-speyer-marcus1982a,krainak-speyer-marcus1982b}), called {\it person-by-person (PbP) optimality}, {\it global or team optimality,  and their  relation},   should play a fundamental role in developing analogous  optimality conditions  for  decentralized stochastic dynamic optimal control problems. Although, 
static team theory is successfully applied to  decentralized stochastic dynamic optimal control problems with   one-step delayed sharing information patterns in \cite{kurtaran-sivan1973,yoshikawa1975,sandell-athans1974,varaiya-walrand1978,aicardi-davoli-minciardi1987,bansal-basar1987}, its  generalization to arbitrary$-$step delayed sharing patterns remains a challenge, especially due to the counterexample of Varaiya and Walrand \cite{varaiya-walrand1978}. Often, alternative methods are considered  under  simplified assumptions 
\cite{waal-vanschuppen2000,bamieh-voulgaris2005,raoufat-djouadi:2018,lessard-lall2011,vanschuppen2011,vanschuppen2012,nayyar-mahajan-teneketzis2013,nedic:2017,ho-chu1972,ho-chu1973,tobias2024}. 
The main limitations of   static team theory are attributed to the assumptions that,  \\
{\it iii) the information pattern  available to  the strategies of the decision makers    are not affected by any of the strategies, and\\
 iv)  there are no state dynamics or the dynamics are not affected by the strategies of the decision makers.  }\\
 To overcome  limitation iii),   Witsenhusen in the  1988 paper   \cite{witsenhausen1988}, considered   a class of problems without  state  dynamics,  with   countable  control and observation spaces. \cite{witsenhausen1988}  
 proved using the so-called  {\it common denominator condition}, that there  exist   a probability measure,  such that under this measure  the observations become an independent process \cite[Section~4]{witsenhausen1988},  and then applied a {\it change of variables} \cite[Section~5]{witsenhausen1988},  to obtain   an equivalent  problem, called {\it the static reduction problem} \cite[Section~6]{witsenhausen1988}. 
Although, very powerful  the static reduction  approach remained unexplored for several years, because it requires  the construction of the common denominator condition, which was not provided   in \cite{witsenhausen1988}.
 
Recently,   \cite{Bambos_equivalence} considered  decentralized stochastic dynamical optimal control problems in state space, and  proved that  the common denominator condition is precisely the   discrete-time version of Girsanov's change of probability measure. \cite{Bambos_equivalence}  constructed using a Radon-Nikodym derivative \cite{liptser-shiryayev1977},  an   equivalent  references measure, such that under this measure,   the  state process  and observation processes are mutually  independent. 
  The 
{\it  change of variables} is equivalent to Bayes' theorem.
Girsanov's  change of measure is applied in  \cite{charalambous-ahmed:IEEEAC2017a,charalambous-ahmed:IEEEAC2018,charalambous:MCSS2016}  to  developed  global and PbP optimality conditions, which  generalize  Radner's theorem of  stationary conditions \cite{radner1962},   to 
 controlled stochastic differential equations (SDEs) with multiple controls having different information patterns. 
   The optimality conditions are  expressed in terms  of  a ``Hamiltonian Systems" consisting of  backward and forward  SDEs. The   optimal  strategies are  determined by  a conditional variational Hamiltonian, conditioned on its  information structure.
\\ Girsanov's change of measure  is also  constructed  in   \cite{Bambos_equivalence}, for   discrete-time decentralized stochastic dynamic optimal control  problems described by  nonlinear state space models (generalizing and completing the construction of  \cite{witsenhausen1988}). 

\subsection{Main Contributions of the Paper}
The  main contribution  is twofold.    
\begin{enumerate}
\item  Derivation of  PbP and global optimality conditions using  Girsanov's change of measure,  for general  discrete-time decentralized stochastic optimal control problems described by arbitrary conditional distributions. 

\item Application of the PbP  optimality conditions to  the counterexample to determine   the  optimal strategies $\gamma^o(\cdot) \tri (\overline{\gamma}_1^o(\cdot), \gamma_2^o(\cdot))$, $\overline{\gamma}_1(x_0)=x_0+\gamma_1(x_0)$,  when $v \in G(0, \sigma^2)$ and $x_0$ has arbitrary   distribution  ${\bf P}_{x_0}$. These are  given by  the conditional expectations, as follows.
\begin{align}
&\overline{\gamma}_1^o(x_0) =x_0 - \frac{1}{k^2}{\bf E}^{\gamma^o}\Big\{ \overline{\gamma}_1^o(x_0) - \gamma_2^o(y_1) \Big|x_0 \Big\} \label{nl_1} \\
& -\frac{1}{2k^2\sigma^2} {\bf E}^{\gamma^o}\Big\{ \big(y_1-\overline{\gamma}_1^o(x_0)\big) \big(\overline{\gamma}_1^o(x_0)  - \gamma_2^o(y_1)\big)^2 \Big|x_0 \Big\}, \nonumber    \\
&\gamma_2^o(y_1)=   {\bf  E}^{\gamma^o} \Big\{ \overline{\gamma}_1^o(x_0)    \Big| y_1 \Big\}, \label{nl_2} \\\
&y_1=\overline{\gamma}_1^o(x_0)+v, \hso  \overline{\gamma}_1(x_0)=x_0+\gamma_1(x_0). 
\end{align}
Equivalently,  $(\overline{\gamma}_1^o(\cdot), \gamma_2^o(\cdot))$ satisfy the  two nonlinear integral equations, 
\begin{align}
  &  \overline{\gamma}_1^o(x_0) = x_0 - \frac{1}{k^2} \int_{-\infty}^\infty \Big\{\frac{ \big(\zeta - \overline{\gamma}_1^o(x_0)\big) \big(\overline{\gamma}_1^o(x_0) - \gamma_2^o(\zeta)\big)^2}{2\sigma^2} \nonumber \\ 
 &     + \big(\overline{\gamma}_1^o(x_0) - \gamma_2^o(\zeta)\big)\Big\} \frac{ \exp({-\frac{(\zeta-\overline{\gamma}_1^o(x_0))^2}{2\sigma^2}})}{ \sqrt{2\pi \sigma^2}} d\zeta, 
    \label{eq:gamma1bar_in}\\
 &   \gamma_2^o(y_1) = \frac{\int_{-\infty}^{\infty} \overline{\gamma}_1^o(\xi) \exp{(-\frac{(y_1-\overline{\gamma}_1^o(\xi))^2}{2\sigma^2})}  {\bf P}_{x_0}(d\xi)}{\int_{-\infty}^{\infty} \exp{(-\frac{(y_1-\overline{\gamma}_1^o(\xi))^2}{2\sigma^2})} {\bf P}_{x_0}(d\xi)}.
   \label{eq:gamma2_in}
\end{align}
\end{enumerate}

\noindent In addition,  we   provide the  following. 

2.1) A fixed point theorem  of the  two nonlinear integral equations    in a function space, establishing existence and uniqueness of  solutions  $(\overline{\gamma}_1^o(\cdot), \gamma_2^o(\cdot))$.

2.2) Evaluation of the optimal strategies $(\overline{\gamma}_1^o(\cdot), \gamma_2^o(\cdot))$ by solving numerically the  two integral equations for different parameters and comparison   to  numerical  evaluations of the  payoff  found in the literature.

\noi The numerical evaluation of the two optimal strategies verifies    the properties of optimal strategies derived in \cite{WitsenhausenOriginal}. It is observed that for some choices of the problem parameters,  linear or affine strategies are indeed optimal\footnote{This observation is consistent with  \cite{WitsenhausenOriginal}, because Theorem~2 in \cite{WitsenhausenOriginal} states that nonlinear strategies outperform affine strategies for certain choices of the problem parameters, and not for all possible choices of parameters.}.

\subsection{Organization}
\label{subsec:organization}
In Section~\ref{discrete} we treat  general  discrete-time decentralized stochastic optimal control problems,  using  Girsanov's change of measure, and we derive global and PbP optimality conditions. 
 In Section \ref{sec::optimalstrategy}, we present  the derivation of the optimal strategies of the counterexample  \eqref{eq:gamma1bar_in},\eqref{eq:gamma2_in}, and the fixed point theorem.  In Section \ref{sect:num-eva} we present the  numerical integration of  \eqref{eq:gamma1bar_in}, \eqref{eq:gamma2_in} for  different sets of parameter values and we compare our optimal payoff to other studies of the literature. 

\subsection{Notation}
\label{sec:nota}
$\mathbb{R} \tri (-\infty,\infty)$, $\mathbb{Z}_+ \tri \{1,2, \ldots\}$, {$\mathbb{Z}_+^n \tri \{1,2, \ldots,n\}$,}  $n \in {\mathbb Z}_+$. 
  Given a set of elements $s^{(K)} \tri \{s^1, s^2, \ldots, s^K\}$, we  define   $s^{-k} \tri s^{(K)} \setminus\{s^k\}$,  the set $s^{(K)}$ minus element $\{s^k\}$.  
 
 $\{({\mathbb X}_t,{\cal B}({\mathbb X }_t))\big| t\in\mathbb{Z}_+^n\}$ denotes  measurable spaces, where ${\mathbb X}_t$ is  confined to complete separable metric space or Polish space, and ${\cal B}({\mathbb X}_t)$ is the Borel $\sigma-$algebras of subsets of ${\mathbb X}_t,  \forall t\in\mathbb{Z}_+^n$. Points in the product space ${\mathbb X}_{1,n}\triangleq{{\prod}_{t\in\mathbb{Z}_+^n}}{\mathbb X}_t$ are denoted by $x_{1,n}\triangleq(x_1,\ldots, x_n)\in{\mathbb X}_{1,n}$, and their restrictions for any $(m,n)\in\mathbb{Z}_+ \times \mathbb{Z}_+$ by $x_{m,n}\triangleq(x_{m},\ldots, x_n)\in{\mathbb X}_{m,n}, n\geq{m}$. Hence,  ${\cal B}({\mathbb X}_{1,n})\triangleq\otimes_{t\in\mathbb{Z}_+^n}{\cal B}({\mathbb X}_t)$ denotes  the $\sigma-$algebra on ${\mathbb X}_{1,n}$ generated by cylinder sets $\{(x_1,\ldots, x_n)\in{\mathbb X}_{1,n}\big|x_j\in{A}_j, A_j\in{\cal B}({\mathbb X}_j),~j\in\mathbb{Z}_+^n\}$. 
 
Given a measurable  space $\big(\Omega, {\cal F}\big)$, we denote the set of probability measures (PMs) ${\mathbb P}$ on $\Omega$  by ${\cal M}(\Omega)$.  Given a sequence of  RVs indexed by subscript $t$,   $x_t:  (\Omega, {\cal F}) \rar  ({\mathbb X}_t, {\cal B}({\mathbb X}_t)), \forall t \in {\mathbb Z}_+^n$,  we denote  by
${\mathbb P}\big\{x_1 \in d\eta_1, \ldots, x_n \in d\eta_n\big\}={\bf P}_{x_{1,n}}(d\eta_{1,n})\equiv {\bf P}(d\eta_{1,n})$ the PM induced by $x_{1,n}$ on $({\mathbb X}_{1,n}, {\cal B}({\mathbb X}_{1,n}))$ (i.e., probability distribution (PD) if ${\mathbb X}_t= {\mathbb R}^k$). Given another  RV, $y: (\Omega, {\cal F}) \rar  ({\mathbb Y}, {\cal B}({\mathbb Y}))$ we define    the conditional PM  of the  RV $y$ conditioned on   $x$ by $
{\bf P}_{y|x}(d\xi| x)\tri {\mathbb P}\big\{y \in d\xi\big|x\big\}
  \equiv {\bf P}(d\xi|x)$, 
where $x$ is replaced by $\eta$ if    the RV $x$ is  fixed, i.e.,   $x=\eta$. The joint PM  of $(x,y)$ is
${\bf P}_{x,y}(d\eta,d\xi)={\bf P}_{y|x}(d\xi|\eta) {\bf P}_x(d\eta)$.


%
%
%
%

\section{Optimality Conditions for Decentralized Stochastic Optimal Control  Problems}
\label{discrete}
In this section we invoke  a version of Girsanov's change of measure to derive optimality conditions for  general   discrete-time decentralized stochastic optimal control problems (that include  \cite{witsenhausen1971}). To ensure our method applies to processes with values in Euclidean spaces, finite state spaces, etc., the model is described  by arbitrary   conditional PMs.

As in   \cite{witsenhausen1971},  
we consider  ${\mathbb Z}_+^M\tri \{1,2, \ldots, M\}$ observation posts collecting information at each  $t \in T_+^n$ and  ${\mathbb Z}_+^K\tri \{1,2, \ldots, K\}$ control stations applying  control actions at each   $t \in T_+^n$.
The decentralized stochastic control problem  is described 
by  the following elements. 

 1) The unobservable state process,   $x_{1,n} \tri (x_1, x_2,  \ldots, x_n)$, $x_t \in {\mathbb X}_t , \forall t \in T_+^{n}\tri \{1, \ldots, n\}$.

2) The   observation processes  at the observation posts,  $y_{1,n}^{m}\tri (y_{1}^{m}, \ldots, y_{n}^{m})$,    $y_t^m\in {\mathbb Y}_t^m$, $\forall t \in T_+^{n}$,   $\forall m \in {\mathbb Z}_+^M$. For the $M-$tuples we use  the superscript notation   $y_t^{(M)} \tri \big(y_t^1, \ldots, y_t^M\big),  \forall t \in T_+^{n}$ and $y_{1,n}^{(M)}\tri (y_{1,n}^{1}, \ldots, y_{1,n}^{M})$.

3)  The control actions applied at the control stations, $u_{1,n}^k \tri (
u_1^k, u_2^k,  \ldots, u_{n}^k)$, $u_t^k\in  {\mathbb A}_t^k$, $\forall t \in T_+^{n}$, $\forall k \in {\mathbb Z}_+^K$.  $u_t^{(K)}\tri \big(u_t^1, \ldots, u_t^K\big)$, $\forall t \in T_+^{n}$ and $u_{1,n}^{(K)}\tri (u_{1,n}^{1}, \ldots, u_{1,n}^{K})$.   

4)   The conditional  probability measure  (PM) of $x_{t+1}$ conditioned on $(x_{1,t}, y_{1,t}^{(M)},  u_{1,t}^{(K)})$, satisfies 
\begin{align}
& {\mathbb P}\Big\{x_{t+1} \in A_{t+1} \big| x_{1,t},y_{1,t}^{(M)},  u_{1,t}^{(K)}\Big\}=
 {\bf P}_{x_{t+1}|x_{t}, u_{t}^{(K)}}(A_{t+1})\nonumber \\
 &=S_{t+1}(A_{t+1}|x_t, u_{t}^{(K)}), \;       A_{t+1} \in {\cal B}({\mathbb X}_{t+1}), \;   \forall t . \label{i-d1}
 \end{align}

 5) The conditional PM of $y_{t}^m$ conditioned on $(x_{1,t},y_{1,t-1}^{(M)}, y_t^{-m}, u_{1,t}^{(K)})$,   satisfies 
\begin{align}
 &{\mathbb P} \Big\{y_{t}^m \in B_t^m \big| x_{1,t},y_{1,t-1}^{(M)},y_t^{-m},  u_{1,t}^{(K)}\Big\}=
 {\bf P}_{y_{t}^m|x_{t}, u_{t}^{(K)}}(B_t^m)\nonumber \\
 &=Q_{t}^m( B_t^m|x_t, u_{t}^{(K)}), \hso     B_t^m \in {\cal B}({\mathbb Y}_{t}^m), \;   \forall  t,   \; \forall m. \label{i-d2}
 \end{align}
From  (\ref{i-d2})   we also have, 
\begin{align}
 &{\mathbb P} \Big\{y_{t}^{(M)} \in B_t^{(M)} \big| x_{1,t},y_{1,t-1}^{(M)},  u_{1,t}^{(K)}\Big\}, \;    B_t^{(M)} \in {\cal B}({\mathbb Y}_{t}^{(M)})\nonumber \\
 &= Q_{t}^{(M)}(B_t^{(M)} |x_t, u_{t}^{(K)}),  \;   \forall  t \label{i-d2-j-1}\\
 &= \prod_{m=1}^M Q_{t}^m(B_t^m |x_t, u_{t}^{(K)}), \;    B_t^m \in {\cal B}({\mathbb Y}_{t}^m). \label{i-d2-j-2}
 \end{align}

6) The  information  patterns and strategies,  which are used to generate the  control actions $u_t^k, \forall t \in T_+^n, \forall k \in {\mathbb Z}_+^K$,  are defined 
as follows. \\
6.1) {\it The Information Patterns} of control  are  specified by two projection operators. \\
i) For each  $(k, t)$, the  projection  of all observations $y_{1,t-1}^{(M)}$ to any of its subset defined by 
\begin{align}
&{\bf \Pi}_t^k (y_{1,t-1}^{(M)})\tri 
\Big\{y_\tau^\mu \big| \tau \subseteq  \{1, \ldots, t-1\},  \\
& \mu \in \kappa^k(t)  \subseteq \{1, \ldots, M\} \Big\}, \forall k \in \{1, \ldots, K\}, \:  \forall t \in \{2, \ldots, n\}.\nonumber 
\end{align}
ii) For each  $(k, t)$, the  projection  of all controls $u_{1,t-1}^{(K)}$ to any of its subset defined by 
\begin{align}
&{\bf \Pi}_t^k (u_{1,t-1}^{(K)})\tri 
\Big\{u_\tau^\kappa \big| \tau \subseteq  \{1, \ldots, t-1\},  \\
& \kappa \in \iota^k(t)  \subseteq \{1, \ldots, K\} \Big\}, \forall k \in \{1, \ldots, K\}, \:  \forall t \in \{2, \ldots, n\}.\nonumber 
\end{align}
The information pattern of  each  control  station $k \in \{1, \ldots, K\}$  at each time $t \in \{1, \ldots, n\}$ is 
\begin{align}
I_t^k \tri & {\bf \Pi}_t^k (y_{1,t-1}^{(M)})\bigcup {\bf \Pi}_t^k (u_{1,t-1}^{(K)}), \hst  I_1^k \tri\emptyset, \\
& \forall k \in \{1, \ldots, K\}, \:  \forall t \in \{2, \ldots, n\}.  \nonumber 
\end{align}
{\it  6.2) Control Strategies} used by  the controls to generate actions    are  Borel measurable  maps $\gamma_t^k(\cdot)$, 
\begin{align}
u_t^k =
 \gamma_t^k(I_t^k), \hso \;   \forall k\in \{1, \ldots, K\}, \;     \forall t\in \{1, \ldots, n\}.   \label{i21}
\end{align}
For  each $k$, such  strategies  are denoted by ${\cal  U}_{1,n}^k$  with  notation,
\begin{align} 
 \gamma_{1,n}^k (\cdot)\tri & (\gamma_1^k(\cdot), \ldots, \gamma_{n}^k(\cdot)) \in {\cal  U}_{1,n}^k \tri \times_{t=1}^{n} {\cal  U}_t^k, \nonumber \\
 \gamma_{1,n}^{(K)}(\cdot) \tri& (\gamma_{1,n}^1(\cdot), \ldots, \gamma_{1,n}^{K}(\cdot)) \in {\cal  U}_{1,n}^{(K)} \tri \times_{k=1}^K {\cal  U}_{1,n}^k. \nonumber
\end{align} 


{\it  The Joint Probability Measure (PM).} For each $n$, we introduce the space ${\mathbb G}^{n}$ of admissible histories,   $
 {\mathbb   G}^n\tri  (\mathbb{A}_{1,n}^{(K)} \times {\mathbb X}_{1,n} \times \mathbb{Y}_{1,n}^{(M)}), \forall n \in {\mathbb Z}_+^{n}
$.
 We equip the space 
  ${\mathbb  G}^n$  with the natural $\sigma$-algebra 
  ${\cal B}({\mathbb  G}^n), \forall n \in {\mathbb Z}_+^{n}$. 
Then we define the joint PM.  ${\bf P}_{y_{1,n}^{(M)}, u_{1,n}^{(K)}, x_{1,n}}$ of $(y_{1,n}^{(M)}, u_{1,n}^{(K)}, x_{1,n})$  on  the  canonical space $\big({\mathbb   G}^n, {\cal  B}({\mathbb  G}^n)\big)$, and  we construct a probability space $\big(\Omega, {\cal F}, {\mathbb P}^{u}\big)$ carrying the RVs $(y_{1,n}^{(M)}, u_{1,n}^{(K)}, x_{1,n})$, as follows.  For Borel sets,  $A_{1,n}\tri \times_{k=1}^n A_k,  A_k \in{\cal B}({\mathbb X}_k) $,  and similarly for Borel sets  $(B_{1,n}^{(M)},  C_{1,n}^{(K)})$, then 
\begin{align}
&{\mathbb P}^{u}\Big\{x_{1,n} \in A_{1,n} , y_{1,n}^{(M)}\in B_{1,n}^{(M)}, u_{1,n}^{(K)}\in C_{1,n}^{(K)} \Big\}\\
&=Q_n^{(M)}(B_{n}^{(M)} \big|x_n, u_n^{(K)})  {I}_{\{u_n^{(K)} \in C_{n}^{(K)}\}}S_n(A_n   \big| x_{n-1}, u_{n-1}^{(K)})     \nonumber \\
&  \ldots   Q_1^{(M)}(B_{1}^{(M)} \big|x_1, u_1^{(K)})     I_{\{u_1^{(K)}\in C_1^{(K)}\}}   S_1(A_1) \nonumber \\
& \mbox{such that (\ref{i-d1})-(\ref{i-d2-j-2}) hold  } \label{jd-o}
\end{align}
where $C_t^{(K)} \in {\cal B}({\mathbb A}_t^{(K)})$ and  $I_{\{u_t^{(K)}\in C_t^{(K)}\}}=1$ if $u_1^{(K)}\in C_t^{(K)}$ and zero  otherwise.

{\it The Average  Payoff or Cost Function.} The average  payoff is given by the expression,  
\begin{align}
J^{{\mathbb P}^u} (u^{(K)}) \tri {\mathbb E}^{{\mathbb P}^{u}} \Big\{ \sum_{t=1}^{n-1} \ell(t, x_t, u_t^{(K)})+ \kappa(n, x_n) \Big\}   \label{i8-cost_a}
\end{align}
where $\ell(\cdot),  \kappa(\cdot)$ are measurable functions. 

We wish to characterize   decentralized team/global and person-by-person  (PbP)  optimality, based on Definition~\ref{def-team}.

 \begin{definition}(Decentralized Global  and PbP Optimality) \\
 \label{def-team}
(1)  Decentralized Global Optimality. The  $K-$tuple of strategies   $\gamma_{1,n}^{(K),o}\tri   (\gamma_{1,n}^{1,o}, \gamma_{1,n}^{2,o}, \ldots, \gamma_{1,n}^{K,o}) \in {\cal  U}_{1,n}^{(K)}$
  is called  {\it decentralized global  optimal}, if it satisfies 
\begin{align}
J^{{\mathbb P}^u}(\gamma_{1,n}^{(K),o} ) \leq  & J^{{\mathbb P}^u}(\gamma_{1,n}^{(K)}), \hso  \forall  \gamma_{1,n}^{(K)} \in {\cal U}_{1,n}^{(K)}. \label{i10}
\end{align}
(2)  Decentralized PbP Optimality. The  $K-$tuple of strategies   $\gamma_{1,n}^{(K),o}\tri   (\gamma_{1,n}^{1,o}, \gamma_{1,n}^{2,o}, \ldots, \gamma_{1,n}^{K,o}) \in {\cal  U}_{1,n}^{(K)}$
  is called  {\it decentralized PbP optimal}, if it satisfies,   $\forall k \in {\mathbb Z}^K$, 
\begin{align}
 {J}^{{\mathbb P}^u}(\gamma_{1,n}^{k,o}, \gamma_{1,n}^{-k,o}) \leq  {J}^{{\mathbb P}^u}(\gamma_{1,n}^{k}, \gamma^{-k,o}), \;
  \forall \gamma_{1,n}^k \in {\cal  U}_{1,n}^{k}. \label{i11}
\end{align}
  \end{definition}

\begin{example} Our formulation includes  recursive models, 
\label{ex:rec}
\begin{align}
&x_{t+1}=f(t,x_t, u_t^{(M)}, w_t), \; \forall t \in { T}_+^{n-1}, \label{NDM-3}\\
& y_{t}^m   = h^m(t,x_t, u_t^k, v_t^m),\hso \forall t \in {T}_+^{n}, \; m=1, \ldots, M \label{NDM-4} 
\end{align}
where  $\{(x_1, w_1, v_1^m, w_2, v_2^m, \ldots, w_{n}, v_n^m)\big| m =1, \ldots, M\}$ are mutually independent RVs with known  PMs\footnote{${\bf P}_{w_{t}}(dw_{t})$ is  the probability of the event $\{w_{t}\in dw_{t}\}$, i.e.,    the RV $w_{t}$  is found in the set $dw_{t}\subset {\mathbb W}_{t}$.}  ${\bf P}_{x_{1}}(dx_{1})$ and 
\begin{align*}
 {\bf P}_{w_{t}}(dw_{t})=\Psi_t(dw_t),   \: {\bf P}_{v_{t}^k}(dv_{t}^k)=\Phi_t^m(dv_t^m),  \: \forall (t, m).
\end{align*} 
The conditional  PMs 
$S_{t+1}(dx_{t+1}|x_t, u_{t}^{(K)}), Q_{t}^m(dy_t^m |x_t, u_{t}^{(K)})$ are determined from the model. 
\end{example}

\begin{remark} Since on probability measure ${\mathbb P}^u$,   $(x_{1,n}, y_{1,n}^{(M)})$ are affected by $u_{1,n}^{(K)}$,   it is almost impossible to apply calculus of variations to $J^{{\mathbb P}^u} (u^{(K)})$. To circumvent this technicality  we invoke a change of measure to  transform $J^{{\mathbb P}^u} (u^{(K)})$ to an equivalent payoff $J^{\sr{\circ}{\mathbb P}} (u^{(K)})$ on a reference  measure $\sr{\circ}{\mathbb  P}$ such that $(x_{1,n}, y_{1,n}^{(M)})$ are not  affected by  $u_{1,n}^{(K)}$. 
We derive optimality conditions on  $\sr{\circ}{\mathbb  P}$ and 
 translate them  on  ${\mathbb  P}^u$. 
  \end{remark}

\subsection{Change of Probability Measures}
\label{sect:ref-original}
The mathematical concept we use to change  the probability measure is known as Radon-Nikodym derivative theorem (see brief summary of the basic theorems   in Section~\ref{app-1}). 

{\it  1)  Change from  Reference Measure $\sr{\circ}{\mathbb  P}$ to the Original  Measure ${\mathbb  P}^u$.}
We consider  a reference   probability space $(\Omega, {\cal  F}, \sr{\circ}{\mathbb  P})$ such that the following hold. 

\indent 1.1)  The processes 
 $x_{1,n}$ are  mutually independent with PM
\begin{align} 
{\bf P}_{x_{1,n}}(dx_{1,n})=\prod_{t=1}^n
  {\bf P}_{x_t}(dx_t) \equiv \prod_{t=1}^n\Psi_t(dx_t).  \label{ref_1}
\end{align}
 \indent 1.2)  The processes  $y_{1,n}^m$  are mutually  independent for each  $ m \in {\mathbb Z}_+^M$,  and $y_{1,n}^k$ is independent of  $y_{1,n}^m$, $\forall k\neq m$,  with PMs, 
 \begin{align} 
 &{\bf P}_{y_{1,n}^m}(dy_{1,n}^m)=\prod_{t=1}^n {\bf P}_{y_t^m}(dy_t^m)\equiv \prod_{t=1}^n\Phi_t^m(dy_t^m), \;  \;  \forall m,   \label{ref_3}   \\
 & {\bf P}_{y_t^{(M)}}(dy_t^{(M)})=\Phi_{t}^{(M)}(dy_t^{(M)})=\prod_{m=1}^M{\bf P}_{y_t^m}(dy_t^m), \; \forall t.  \label{ref_4}
\end{align}
1.3)  The processes $x_{1,n}$ and  $y_{1,n}^{(M)}$ are independent, i.e., 
\begin{align}
{\bf P}_{x_{1,n}, y_{1,n}^{(M)}}(dx_{1,n}, dy_{1,n}^{(M)} )= {\bf P}_{x_{1,n}}(dx_{1,n}){\bf P}_{y_{1,n}^{(M)}}( dy_{1,n}^{(M)}).  \label{ref_4_new}
\end{align}
We introduce the $\sigma-$algebars generated  by the indicated RVs. 
 \begin{align*}
&{\cal F}_{t}^{0, {x}}\tri \sigma\big\{(x_1, \ldots, x_{t}) \big\}, \; \forall t \in T_+^n, \\
&{\cal F}_{t}^{0, {y^m}}\tri \sigma\big\{(y_1^m, \ldots, y_{t}^m) \big\}, \; \forall m, \\
&{\cal F}_{t}^{0, y^{(M)}, u^{(K)}}\tri \sigma\big\{(y_1^{(M)}, \ldots, y_{t}^{(M)},  u_{1}^{(K)}, \ldots, u_{t}^{(K)}) \big\}, \\
& {\cal F}_{t}^0 \tri \sigma\big\{(x_{1, t}, y_{1,t}^{(M)},  u_{1,t}^{(K)}) \big\}, \hso 
{\cal F}^{0,I_t^k} \tri \sigma\big\{I_t^k \big\}, \hso \forall k.
\end{align*}
Let $\{{\cal F}_{t}^x\big|  t \in T_+^n  \}$, $\{{\cal F}_{t}^{{y^k}}\big|  t \in T_+^n  \}$,  
$\{{\cal F}_{t}^{y^{(M)}, u^{(K)}} \big|  t \in T_+^n  \}$, $\{{\cal F}_{t}\big|  t \in T_+^n  \}$  denote the    complete filtrations  generated  by the $\sigma-$algebras   ${\cal F}_{t}^{0, {x}}, {\cal F}_{t}^{0, {y}^k}$, 
 ${\cal F}_{t}^{0, y^{(M)}, u^{(K)}}, {\cal F}_{t}^0$,  respectively, and let ${\cal F}^{I_t^k}$ denote the $\sigma-$algebra generated by the information pattern $I_t^k$, all  augmented by the $\sr{\circ}{\mathbb P}-$null sets of ${\cal F}$. 
 

Starting with  the reference   probability space $(\Omega, {\cal  F}, \sr{\circ}{\mathbb  P})$ such that 1.1)-1.3) hold,   we will  construct   a   probability space $(\Omega, {\cal  F}, {\mathbb  P}^{u})$ using the Radon-Nykodym derivative Theorem~\ref{theorem 1.7.4},   $\frac{d{\mathbb  P}^u}{d\sr{\circ}{\mathbb  P}}\Big|_{{\cal F}_n}\tri \Lambda_n^u M_n^u$, where   $\Lambda_n^u M_n^u$ is appropriately chosen so that  
 the following hold.
%
%

\indent 1.4) On the probability space $(\Omega, {\cal  F}, {\mathbb  P}^u)$ the conditional PMs of $\big\{x_t\big| t \in T_+^n\big\}$ and $\big\{y_t^k\big| t \in T_+^n\big\}, \forall k $ are given by (\ref{i-d1})-(\ref{i-d2-j-2}) 
and the joint probability distribution is (\ref{jd-o}).

\begin{theorem} Change from a Reference  Measure $\sr{\circ}{\mathbb  P}$ to the Original Measure ${\mathbb  P}^u$.\\
\label{thm:RND_1}
 Consider the reference probability space  $(\Omega, {\cal  F}_{n}, \sr{\circ}{\mathbb  P})$ on which  1.1)-1.3) hold, i.e., $x_{1,n}\tri \big\{x_t\big| t \in T_+^n\big\}$ and $y_{1,n}^m \tri \big\{y_t^m\big| t \in T_+^n\big\}$,    $\forall  m \in {\mathbb Z}_+^M$ are independent with 
   PMs (\ref{ref_1})-(\ref{ref_4_new}). \\
\noi Define the  processes,  
\begin{align}
&\lambda_s^u  \tri   \frac{   Q_{s}^{(M)}(dy_{s}^{(M)} |x_s, u_{s}^{(K)})}{\Phi_s^{(M)}(dy_s^{(M)})}  = \prod_{m=1}^M \frac{   Q_{s}^{m}(dy_{s}^{m} |x_s, u_{s}^{(K)})}{\Phi_s^{m}(dy_s^{m})},\nonumber 
\\
&m_s^u \tri   \frac{S_s(dx_s|x_{s-1}, u_{s-1}^{(K)})}{\Psi_s(dx_s)},\; m_1^u=1,   \hso s=1, \ldots, n,  \label{RND_1.1-l}\\
&\Lambda_t^u \tri  \prod_{s=1}^t  \lambda_s^u, \;  M_t^u\tri  \prod_{s=1}^t  m_s^u,  \; M_1^u =1,  \;  \forall t \in {T}_+^n. \label{RND_1.1} 
\end{align}
Assume $Q_{s}^{(M)}(\cdot |x_s, u_{s}^{(K)})$ is absolutely continuous w.r.t $\Phi_s^{(M)}(\cdot)$, denoted by  $Q_{s}^{(M)}(\cdot |x_s, u_{s}^{(K)}) \ll \Phi_s^{(M)}(\cdot)$,   for almost all $(x_s, u_{s}^{(K)})$, and $S_{s}(\cdot |x_{s-1}, u_{s-1}^{(K)}) \ll\Psi_s(\cdot)$ for almost all $(x_{s-1}, u_{s-1}^{(K)})$, $\forall s$,  and $\{\Lambda_t^u M^u_t\big| t \in T_+^n\}$ is $\sr{\circ}{\mathbb P}-$integrable. \\
The following hold. \\
(1) The process  $\{\Lambda_t^u M_t^u\big| t \in T_+^n\}$ is an   $\big(\{{\cal F}_{t}\big| t\in T_+^n\}, \sr{\circ}{\mathbb P}\big)$ martingale\footnote{The martingale is defined by,    $\Lambda_t^u M_t^u $ is ${\cal F}_{t}-$measurable,   $\Lambda_t^u M_t^u $  is $\sr{\circ}{\mathbb P}-$integrable, and  (\ref{RND_O_1_d_l_1}) holds.}, i.e., 
\begin{align}
&{\bf  E}^{ \sr{\circ}{\mathbb P }}\Big\{\big(\Lambda_t^u M_t^u \big|{\cal F}_{ t-1}\Big\} =\Lambda_{t-1}^uM_{t-1}^u,  \;
 \forall t \in  T_+^n , \label{RND_O_1_d_l_1}\\
&{\bf  E}^{\sr{\circ}{\mathbb P}}\Big\{\Lambda_t^u M_t^u\Big\}=1, \hso \forall t \in T_+^n. \label{rnd_1_1}
\end{align}
(2)  Define, 
\begin{align}
&\frac{d{\mathbb  P}^u}{d\sr{\circ}{\mathbb P}}\Big|_{{\cal F}_t}\tri \Lambda_t^u M_t^u, \hso \forall t \in T_+^n.  \label{RND_1_thm_a_1_1}  
\end{align}
Then,  ${\mathbb P}^u \ll \sr{\circ}{\mathbb  P}$ and 
\begin{align}
& {\mathbb  P}^u(B)=\int_B   \Lambda_t^u(\omega) M_t^u(\omega) d\sr{\circ}{\mathbb P}(\omega), \; \forall 
B\in {\cal  F}_t . \label{RND_1_thm_b_1}
\end{align}
is a probability measure.\\
(3) On the probability space $(\Omega, \{{\cal  F}_t  | t \in T_+^n\}, {\mathbb  P}^u)$, with ${\mathbb  P}^u(B)=\int_B   \Lambda_t^u(\omega) M_t^u(\omega) d\sr{\circ}{\mathbb P}(\omega),  \forall 
B\in {\cal  F}_t$,   the conditional PMs of $ \big\{X_t\big| t \in T_+^n\big\}$ and $ \big\{Y_t^{(M)}\big| t \in T_+^n\big\}$ are given by 
(\ref{i-d1})-(\ref{i-d2-j-2}) and the joint PM is (\ref{jd-o}).
\end{theorem}
\begin{proof} (1) First, we show $\{\Lambda_t^u M_t^u \big| t \in T_+^n\}$ is an   $\big(\{{\cal F}_t\big| t\in T_+^n\}, \sr{\circ}{\mathbb P}\big)$ martingale, by considering $\forall t \in \{1,\ldots, n\}$, 
\begin{align}
& {\bf  E}^{\sr{\circ}{\mathbb P }}\Big\{\Lambda_t^u M_t^u \big|{\cal F}_{t-1}\Big\}
=\Lambda_{t-1}^u M_{t-1}^u   {\bf E}^{\sr{\circ}{\mathbb P}}
\Big\{ \lambda_t^u m_t^u\big|{\cal F}_{t-1}\Big\} 
\label{RND_O_1_a_1}\\
&= \Lambda_{t-1}^u M_{t-1}^u {\bf E}^{\sr{\circ}{\mathbb P}}
\Big\{  {\bf E}^{\sr{\circ}{\mathbb P}}
\Big\{ \lambda_t^u m_t^u \big|{\cal F}_{t-1}, x_t, u_t^{(K)}\Big\} \Big\} \nonumber \\
&=\Lambda_{t-1}^u M_{t-1}^u  {\bf E}^{\sr{\circ}{\mathbb P}}
\Big\{ m_t^u  \; {\bf E}^{\sr{\circ}{\mathbb P}}
\Big\{ \lambda_t^u  \big|{\cal F}_{t-1}, x_t, u_t^{(K)}\Big\} \big|  {\cal F}_{t-1}  \Big\}    \label{RND_O_1_b_1}
\end{align}
where (\ref{RND_O_1_a_1}) is due to $\Lambda_{t-1}^u M_{t-1}^u$ is ${\cal F}_{t-1}-$measurable, and 
 (\ref{RND_O_1_b_1})    is due to $m_t$ is $({\cal F}_{t-1}, x_t, u_t^{(K)})-$measurable. 
We compute the  inner conditional expectation  in  (\ref{RND_O_1_b_1}),  using the fact that    under measure ${\sr{\circ}{\mathbb P}}$,     $y_t^{(M)}$  is independent of $(x_t, u_t^{(K)})$ and  ${\cal F}_{t-1}$:
 \begin{align}
 &{\bf E}^{\sr{\circ}{\mathbb P}}
\Big\{\lambda_t^u \big|{\cal F}_{t-1}, x_t, u_t^{(K)}\Big\}\label{der-1}\\
&=     \int  \frac{Q_{t}^{(M)}(dy_t^{(M)}|x_t, u_t^{(K)}))}{\Phi_{t}^{(M)}(dy_t^{(M)})}\Phi_{t}^{(M)}(dy_t^{(M)})  \nonumber   \\
&{=} \int   Q_{t}^{(M)}(dy_t^{(M)}|x_t, u_t^{(K)}))  =1-a.s.,  \hso \forall t.  \label{RND_O_1a1}\\
& {\bf E}^{\sr{\circ}{\mathbb P}}
\Big\{ m_t^u   \Big|  {\cal F}_{t-1}  \Big\}={\bf E}^{\sr{\circ}{\mathbb P}}
\Big\{\frac{S_t(dx_{t}|x_{t-1}, u_{t-1}^{(K)})}{\Psi_t(dx_t)} \big|  {\cal F}_{t-1}  \Big\}  \\
&=\int  \frac{S_t(dx_{t}|x_{t-1}, u_{t-1}^{(K)})}{\Psi_t(dx_t)}\Psi_t(dx_t) =1-a.s.   \; \forall t.   \label{RND_O_1a1-new-0}
\end{align}
Substituting   (\ref{RND_O_1a1}) ,   (\ref{RND_O_1a1-new-0}) 
into (\ref{RND_O_1_b_1}) 
we
 the martingale property, 
\begin{align}
{\bf  E}^{\sr{\circ}{\mathbb P }}\Big\{\Lambda_t^u M_{t}^u \big|{\cal F}_{t-1}\Big\} =\Lambda_{t-1}M_{t-1}-a.s.,  \;
 \forall t .\label{RND_O_1_d_1}
\end{align}
Taking expectation of both sides of  (\ref{RND_O_1_d_1}) we obtain ${\bf  E}^{\sr{\circ}{\mathbb P}}\big\{\Lambda_{t}^u M_{t}^u\big\}={\bf  E}^{\sr{\circ}{\mathbb P}}\big\{\Lambda_{t-1}^uM_{t-1}^u\big\},  \forall t$, 
which  also implies 
$
{\bf  E}^{\sr{\circ}{\mathbb P}}\big\{\Lambda_{t}^uM_{t}^u\big\}=1,  \forall t \in \{1, \ldots, n\}$,  
thus establishing (\ref{rnd_1_1}). This completes the proof of  the statements under (1). (2) The statements under  (2) follow  by defining  the Radon-Nikodym derivative,   $\frac{d{\mathbb  P}^u}{d{\sr{\circ}{\mathbb  P}}}\Big|_{{\cal F}_t}\tri \Lambda_t M_t,  \forall t \in T_+^n$,  and  using   (\ref{rnd_1_1})  (see    Theorem~\ref{theorem 1.7.4}). (3) 
Consider the bounded continuous functions with compact support,  $\psi:{\mathbb X}_t \rightarrow {\mathbb R}, 
\phi^{(M)}: {\mathbb Y}^{(M)}\rightarrow {\mathbb R}$, and $\phi^{(M)}(y^{(M)})\tri \prod_{m=1}^M \phi^m(y^m)$. 
 It suffices to  show, 
 \begin{align}
&{\mathbb E}^{{\mathbb P}^{u}}\Big\{ \psi(x_t)\phi^{(M)}(y_t^{(M)})\Big|{\cal F}_{t-1}\Big\}=\int \psi(x_t)S_t(dx_t|x_{t-1}, u_{t-1}^{(K)}) \nonumber\\
&\hspace{1.cm} \times~ \prod_{m=1}^M   \phi^{m}(y_t^m)Q_t^{m}(dy_t^m|x_t, u_t^{(K)}), \hso \forall t. \label{der-iid}
\end{align}
To show  (\ref{der-iid}) we use  with   Bayes' rule (see Theorem~\ref{thm:RND-c}.(2)), 
\begin{align}
&{\mathbb E}^{{\mathbb P}^{u}}\Big\{\psi(x_t)\phi^{(M)}(y_t^{(M)})\Big|{\cal F}_{t-1}\Big\}, \hst \forall t \nonumber \\
&=\frac{{\mathbb E}^{\sr{\circ}{\mathbb  P}}\Big\{  \psi(x_t)\phi^{(M)}(y_t^{(M)})   \Lambda_t^u M_t^u \Big|{\cal F}_{t-1}\Big\}}{{\mathbb E}^{\sr{\circ}{\mathbb  P}}\Big\{\Lambda_t^u M_t^u \Big|{\cal F}_{t-1}\Big\}}=\mbox{ (\ref{der-iid})}
\end{align}
where the last equality is shown  similar  to (1). 
\end{proof}

\begin{remark}
\label{rem:rem_2}
We emphasize that on the reference measure $\sr{\circ}{\mathbb P}$,

i) the PMs  induced by $(x_{1,n}, y_{1,n}^{(M)})$ do not depend on $u_{1,n}^{(K)}$, 

ii) the filtrations $\{{\cal F}_{t}^{{x}}\big|  t \in T_+^n  \}$ and       $\{{\cal F}_{t}^{{y^m}}\big|  t \in T_+^n  \}, \forall m \in {\mathbb Z}_+^M$ do note depend on $u_{1,n}^{(K)}$, and 
 
 iii)  for each $t$,  control  $u_t^k=\gamma_t(I_t^k)$ is ${\cal F}^{I_t^k}-$measurable, and ${\cal F}^{I_t^k} \subseteq {\cal F}_{t-1}^{y^{(M)}, u^{(K)}}$, $\forall k$. 
 
  If  only a change of measure on   $y_{1,n}^{(M)}$  is considered, then  $\frac{d{\mathbb  P}^u}{d \sr{\circ}{\mathbb  P}}\Big|_{{\cal F}_t}\tri \Lambda_t^u, \forall t \in T_+^n$,  the conditional  PM of $x_{1,n} $ is the same under both ${\mathbb  P}^u$ and $\sr{\circ}{\mathbb  P}$, i.e., (\ref{i-d1}) holds, and  on  $\sr{\circ}{\mathbb  P}$,  $x_{1,n}$ and $y_{1,n}^{(M)} $,     are independent, and  (\ref{ref_4}) holds. 
 
 
\end{remark}

{\it  2) Reverse Change  from the Original  Measure ${\mathbb  P}^u$ to the Reference  Measure $\sr{\circ}{\mathbb  P}$.}\\
We can also start with the original probability measure  ${\mathbb P}^{u}$ such that  
(\ref{i-d1})-(\ref{i-d2-j-2}) and   (\ref{jd-o}) hold
and  define the reference probability measure  $\sr{\circ}{\mathbb P}$ such that (\ref{ref_1})-(\ref{ref_4_new}) hold, as follows.

Consider 
 $(\Omega, \{{\cal  F}_t|t \in T_+^n\}, {\mathbb  P}^u)$ such that the following hold. 

\indent 2.1)  The conditional PMs of  
 $x_{1,n}\tri \big\{x_t\big| t \in T_+^n\big\}$ and $y_{1,n}^m \tri \big\{y_t^m\big| t \in T_+^n\big\}$,  $\forall  m \in {\mathbb Z}_+^M$ are given by (\ref{i-d1})-(\ref{i-d2-j-2}).
 
%
\noi  Then we can  construct  the reference   probability space $(\Omega, \{ {\cal  F}_t|t \in T_+^n \}, \sr{\circ}{\mathbb  P})$,   such that  $\sr{\circ}{\mathbb  P} \ll {\mathbb  P}^{u}$, by setting 
\begin{align} 
 \frac{d\sr{\circ}{\mathbb  P}}{d{\mathbb  P}^{u}}\Big|_{{\cal F}_t}\tri \big(\Lambda_t^u\big)^{-1} \big(M_t^u\big)^{-1},  \hso \forall t \in T_+^n  \label{rev-rnd}
\end{align} 
 where ${\Lambda}_t, {M}_t^u$ are  define before. We can show (similar to the proof of Theorem~\ref{thm:RND_1}) that  the following hold.

2.2) Under the reference   probability space $(\Omega, \{ {\cal  F}_t|t \in T_+^n \}, \sr{\circ}{\mathbb  P})$ the statements  1.1)-1.3) hold, i.e., (\ref{ref_1})-(\ref{ref_4_new}).

\ \

{\it 3) Equivalent Payoffs.}
Now, we use Theorem~\ref{thm:RND_1}, i.e.,  the Radon-Nikodym derivative  $\frac{d{\mathbb  P}^u}{d\sr{\circ}{\mathbb  P}}\Big|_{{\cal F}_t}\tri \Lambda_t^u M_t^u,  \forall t \in T_+^n$,      to equivalently express  the  payoff  $J^{{\mathbb P}^u}(\gamma_{1,n}^{1} \ldots, \gamma_{1,n}^{K})$ of Definition~\ref{def-team},  under the reference probability measure $\sr{\circ}{\mathbb  P}$.  


\begin{theorem} Equivalent Payoffs \\
\label{thm:p-off}
Define the payoff  on   probability space $(\Omega, \{{\cal  F}_t|t\in T_+^n\}, {\mathbb  P}^u)$,
\begin{align}
&{\mathbb  P}^{u}: \: J^{{\mathbb P}^u}(u^{(K)}) \tri   {\bf E}^{{\mathbb P}^u} \Big\{ \sum_{t=1}^{n-1} \ell(t, x_t, u_t^{(K)})\nonumber \\
&+ \kappa(n, x_n) \Big\}, \; \mbox{s.t.  $(x_{1,n}, y_{1,n}^{(M)})$,    satisfy  (\ref{i-d1})-(\ref{i-d2-j-2})}.  \label{i8-a}
\end{align}
Define the payoff  on reference  probability space $(\Omega, {\cal  F}, \sr{\circ}{\mathbb  P})$,
\begin{align}
&\sr{\circ}{\mathbb P}: \hso J^{\sr{\circ}{\mathbb P}}(u^{(K)})\tri {\bf E}^{\sr{\circ}{\mathbb P}} \Big\{ \sum_{t=1}^{n-1} \ell(t, x_t, u_t^{(K)})\Lambda_t^u M_t^u \nonumber \\
&+ \kappa(n, X_n)\; \Lambda_n^u M_n^u \Big\}, \; \mbox{ $\frac{d{\mathbb  P}^u}{d
\sr{\circ}{\mathbb  P}}\Big|_{{\cal F}_n}\tri \Lambda_n^u M_n^u$ of Thm~\ref{thm:RND_1}},
 \label{i8-or_1}\\
&\mbox{s.t.   $(x_{1,n}, y_{1,n}^{(M)})$
   satisfy  (\ref{ref_1})-(\ref{ref_4_new})}. \nonumber 
\end{align}
Then  the two payoffs are equal, i.e.,  $J^{{\mathbb P}^u}(u^{(K)}) =J^{\sr{\circ}{\mathbb P}}(u^{(K)})$.
\end{theorem} 
\begin{proof} 
Suppose we start  with  $(\Omega, \{{\cal  F}_t| t \in T_+^n\}, {\mathbb  P}^{u})$ on which the  payoff  is $J^{{\mathbb P}^u}(u^{(K)}) =\mbox{(\ref{i8-a})}$. By  Theorem~\ref{thm:RND_1}, and using     $\frac{d{\mathbb  P}^u}{d\sr{\circ}{\mathbb  P}}\Big|_{{\cal F}_n}\tri \Lambda_n^u M_n^u,  \forall n$,  we have the following. 
\begin{align}
&J^{{\mathbb P}^{u}}(u^{(K)}) 
={\bf E}^{\sr{\circ}{\mathbb P}} \Big\{\Lambda_n^u M_n^u \Big( \sum_{t=1}^{n-1} \ell(t, x_t, u_t^{(K)}) 
+ \kappa(n, x_n)\Big) \Big\}.
  \label{i8-or_6}\\
&{\bf E}^{\sr{\circ}{\mathbb P}} \Big\{\Lambda_n^u M_n^u  \sum_{t=1}^{n-1} \ell(t, x_t, u_t^{(K)}) \Big\}={\bf E}^{\sr{\circ}{\mathbb P}} \Big\{ \sum_{t=1}^{n-1} \ell(t, x_t, u_t^{(K)}) \Lambda_n^u M_n^u\Big\}\nonumber \\
&={\bf E}^{\sr{\circ}{\mathbb P}} \Big\{ \sum_{t=1}^{n-1} {\bf E}^{\sr{\circ}{\mathbb P}} \Big\{\ell(t, x_t, u_t^{(K)}) \Lambda_n^u M_n^u\Big|{\cal F}_t \Big\} \Big\} \label{i8-or-3}\\
&\sr{(a)}{=}{\bf E}^{\sr{\circ}{\mathbb P}} \Big\{ \sum_{t=1}^{n-1}\ell(t, x_t, u_t^{(K)}) {\bf E}^{\sr{\circ}{\mathbb P}} \Big\{ \Lambda_n^u M_n^u\Big|{\cal F}_t \Big\} \Big\} \label{i8-or-3}\\
&\sr{(b)}{=}{\bf E}^{\sr{\circ}{\mathbb P}} \Big\{ \sum_{t=1}^{n-1}\ell(t, x_t, u_t^{(K)}) \Lambda_t^u M_t^u\Big\} \label{i8-or-4}
\end{align}
where $(a)$ is due to $\ell(t, x_t, u_t^{(K)})$ is ${\cal F}_t-$measurable, and $(b)$ is  due to Theorem~\ref{thm:RND_1}.(1), $\{\Lambda_t^u M_t^u \big|t \in T_+^n\}$ is an $\big(\{ {\cal F}_t\big|t \in T_+^n\}, \sr{\circ}{\mathbb P}\big)-$martingale. Substituting (\ref{i8-or-4}) into   (\ref{i8-or_6}) we obtain (\ref{i8-a}). Similarly, we can start  with  $(\Omega, \{{\cal  F}_t|t \in T_+^n\}, \sr{\circ}{\mathbb  P})$ on which the payoff  is  $J^{\sr{\circ}{\mathbb P}}(u^{(K)})=\mbox{(\ref{i8-or_1})}$,  and   
show $J^{\sr{\circ}{\mathbb P}}(u^{(K)})=J^{{\mathbb P}^{u}}(u^{(K)})$. 
\end{proof}

\subsection{Conditions for  Global and PbP Optimality  }
By
  Theorem~\ref{thm:p-off}, on  the reference  probability space $(\Omega, {\cal  F}, \sr{\circ}{\mathbb  P})$,  the state and observations $(x_{1,n}, y_{1,n}^{(M)})$  are not affected by the controls $u_{1,n}^{(K)}$. Over the time horizon $\{1, \ldots, n\}$ there are $n \times K$ controls,  $\big\{u_t^k| (t,k)\in \{1, \ldots, n\}\times \{1, \ldots, K\}\big\}$. 
  
  In Theorem~\ref{st-ksm},  we derive stationary conditions for PbP and global optimality  on the reference measure  $ \sr{\circ}{\mathbb  P}$, using concepts from static team theory, and then transform these on the original measure $ {\mathbb  P}^u$.



\begin{theorem} (Stationary Conditions of Decentralized Stochastic Dynamic Optimal Control Problems)\\
\label{st-ksm}
Consider Definition~\ref{def-team} of  decentralized team or global and PbP  optimality. 
Define the sample payoff on the reference  probability space $(\Omega, {\cal  F}, \sr{\circ}{\mathbb  P})$ by 
\begin{align}
L&(x_{1,n}, y_{1,n}^{(M)}, u_{1,n}^{(K)})\tri \sum_{t=1}^{n-1} \ell(t, x_t, u_t^{(K)})   \Theta_t^u \nonumber \\
& + \kappa(n, X_n)\;\Theta_n^u,  \hso  \Theta_t^u \tri \Lambda_t^u M_t^u .
\end{align}
Introduce, $(\nabla_{z_1},\ldots,  \nabla_{z_k})\tri (\frac{\partial }{\partial z_1},\ldots,  \frac{\partial }{\partial z_k})$.
 Assume the following conditions hold.
 
 { (A1)} $L: {\mathbb X}_{1,n} \times  {\mathbb A}_{1,n}^{(K)} \times {\mathbb Y}_{1,n}^{(M)} \rar {\mathbb R}$ is Borel measurable. 
 
 (A2) $Q_{s}^{(M)}(\cdot |x_s, u_{s}^{(K)}) \ll \Phi_s^{(K)}(\cdot)$,   for almost all $(x_s, u_{s}^{(K)})$, and $S_{s}(\cdot |x_{s-1}, u_{s-1}^{(K)}) \ll\Psi_s(\cdot)$ for almost all $(x_{s-1}, u_{s-1}^{(K)})$, $\forall s$,  and $\{\Lambda_t^u M^u_t\big| t \in T_+^n\}$ is $\sr{\circ}{\mathbb P}-$integrable.
 
 { (A3)} There exists  a PbP optimal strategy  $\gamma_{1,n}^{o, (K)}  \in {\cal U}_{1,n}^{(K)}$ with  $J(\gamma_{1,n}^{(K),o})\tri \inf \big\{  J(\gamma_{1,n}^{(K)})\big|\gamma_{1,n}  \in {\cal  U}_{1,n}^{(K)}\big\}\in (-\infty, \infty)$.
 
 (A4)   $\forall k\in  {\mathbb Z}_+^K$, the Gateaux derivative of $L(x_{1,n}, y_{1,n}^{(M)}, \gamma_{1,n}^{-k,o}, \gamma_{1,n}^{k})$ at $\gamma_{1,n}^{k,o}  \in {\cal  U}_{1,n}^{k}$ in the direction      $\gamma_{1,n}^k -\gamma_{1,n}^{k,o}   \in {\cal  U}_{1,n}^{k}$ exists, and  $\gamma_{1,n}^{k,o} + \veps (\gamma_{1,n}^{k} -\gamma_{1,n}^{k,o} )  \in {\cal  U}_{1,n}^{k}, \forall \veps\in  [0,1]$.\\
The following hold.\\
 (1) If 
 $\gamma_{1,n}^{o, (K)}  \in {\cal U}_{1,n}^{(K)}$ is  PbP optimal then necessarily the following stationary conditions hold. \\
 (1.1)  Under the reference measure $ \sr{\circ}{\mathbb P}$, 
\begin{align}
&  {\bf  E}^{\sr{\circ}{\mathbb P}}\Big\{   \nabla_{u_{1,n}^k}L(x_{1,n}, y_{1,n}^{(M)},\gamma_{1,n}^{-k,o},u_{1,n}^{k})\big|_{u_{1,n}^{k}=\gamma_{1,n}^{k,o}}\nonumber \\
&. \Big(\gamma_{1,n}^{k} -\gamma_{1,n}^{k,o}\Big) \Big\}\geq 0,   \hso  \forall \gamma_{1,n}^{k}\in {\cal U}_{1,n}^{k}, \hso \forall k \in {\mathbb Z}_+^K,  \label{vaeqst_b}\\
& {\bf  E}^{\sr{\circ}{\mathbb P}} \Big\{\sum_{t=1}^n  \nabla_{u_t^k} L(x_{1,n}, y_{1,n}^{(M)},\gamma_{1,n}^{-k,o}, \gamma_{1,t-1}^{o}, u_{t}^{k},\gamma_{t+1,n}^{o})\big|_{u_t^k=\gamma_t^{k,o}}\nonumber \\
&. \big( \gamma_t^k-\gamma_t^{k,o}\big) \Big\} \geq 0, \;  \forall  \gamma_{t}^k  \in {\cal  U}_{t}^{k}, \hso \forall (k,t) \in {\mathbb Z}_+^K \times T_+^n.  \label{vaeqst_c}
 \end{align}
 Moreover, the conditional stationary condition holds, 
  \begin{align}
&{\bf  E}^{\sr{\circ}{\mathbb P}} \Big\{  \nabla_{u_t^k} L(x_{1,n}, y_{1,n}^{(M)},\gamma_{1,n}^{-k,o}, \gamma_{1,t-1}^{o},u_{t}^{k},\gamma_{t+1,n}^{o})\big|_{u_t^k=\gamma_t^{k,o}}\nonumber \\
&. \big( \gamma_t^k-\gamma_t^{k,o}\big)\Big| {\cal F}^{I_t^k} \Big\} \geq 0, \;  \forall  \gamma_{t}^k  \in {\cal  U}_{t}^{k}, \;  \sr{\circ}{\mathbb P}\big|_{{\cal F}^{I_t^k}}, \;  \forall (k,t).  \label{vaeqst_con}
 \end{align}
 (1.2)  Under the original  measure ${\mathbb P}^u={\mathbb P}^{\gamma^{(K),o}}$, 
\begin{align}
&  {\bf  E}^{{\mathbb P}^{\gamma^{(K),o}}}\Big\{ \big(\Theta_n^{\gamma^{(K),o}} \big)^{-1} \nabla_{u_{1,n}^k}L(x_{1,n}, y_{1,n}^{(M)},\gamma_{1,n}^{-k,o},u_{1,n}^{k})\big|_{u_{1,n}^{k}=\gamma_{1,n}^{k,o}}\nonumber \\
&.  \Big(\gamma_{1,n}^{k} -\gamma_{1,n}^{k,o}\Big) \Big\}\geq 0,   \;  \forall \gamma_{1,n}^{k}\in {\cal U}_{1,n}^{k}, \; \forall k \in {\mathbb Z}_+^K,  \label{vaeqst_b_orig}\\
& {\bf  E}^{{\mathbb P}^{\gamma^{(K),o}}} \Big\{\sum_{t=1}^n  \nabla_{u_t^k} L(x_{1,n}, y_{1,n}^{(M)},\gamma_{1,n}^{-k,o}, \gamma_{1,t-1}^{o}, u_{t}^{k},\gamma_{t+1,n}^{o})\big|_{u_t^k=\gamma_t^{k,o}}\nonumber \\
&.\big(\Theta_t^{\gamma^{(K),o}}\big)^{-1} \big( \gamma_t^k-\gamma_t^{k,o}\big) \Big\} \geq 0, \;  \forall  \gamma_{t}^k  \in {\cal  U}_{t}^{k}, \; \forall (k,t).  \label{vaeqst_c_orig}
 \end{align}
 Moreover, the conditional stationary condition holds, 
  \begin{align}
&{\bf E}^{{\mathbb P}^{\gamma^{(K),o}}} \Big\{  \nabla_{u_t^k} L(x_{1,n}, y_{1,n}^{(M)},\gamma_{1,n}^{-k,o}, \gamma_{1,t-1}^{o},u_{t}^{k},\gamma_{t+1,n}^{o})\big|_{u_t^k=\gamma_t^{k,o}}\nonumber \\
&.\big( \Theta_t^{\gamma^{(K),o}}\big)^{-1}\big( \gamma_t^k-\gamma_t^{k,o}\big)\Big| {\cal F}^{I_t^k} \Big\} \geq 0, \hso  \forall  \gamma_{t}^k  \in {\cal  U}_{t}^{k},  {\mathbb P}^{\gamma^{(K),o}}\big|_{{\cal F}^{I_t^k}},  \nonumber \\
& \forall (k,t).  \label{vaeqst_con_orig}
 \end{align}
(2) Suppose the following  additional condition holds. \\ 
 {(A5)} $L(x_{1,n}, y_{1,n}^{(M)}, \cdot)$ is convex in $u_{1,n}^{(K)} \in {\mathbb A}_{1,n}^{(K)}$.  \\
 Then any  $\gamma_{1,n}^{o, (K)}  \in {\cal U}_{1,n}^{(K)}$ that satisfies the PbP stationary conditions   (\ref{vaeqst_con}) is also team or globally optimal. 
\end{theorem}
\begin{proof} (1) 
First we show (1.1). 
 Suppose   
$\gamma_{1,n}^{o, (K)}  \in {\cal U}_{1,n}^{(K)}$ is  PbP optimal. For any $\varepsilon \in [0,1]$, define $\gamma_{1,n}^{k,\veps} \tri  \gamma_{1,n}^{k,o} + \veps \big(\gamma_{1,n}^{k} -\gamma_{1,n}^{k,o} \big)  \in {\cal  U}_{1,n}^{k}, k=1, \ldots, K$. Then
we  have
 \begin{eqnarray}
J^{\sr{\circ}{\mathbb P}}(\gamma_{1,n}^{-k,o}, \gamma_{1,n}^{k,\veps})-J^{\sr{\circ}{\mathbb P}}(\gamma_{1,n}^{-k,o}, \gamma_{1,n}^{k,o}) \geq 0, \hso    \forall \varepsilon \in [0,1]. 
 \end{eqnarray}
The  G\^{a}teaux differential  of $J^{\sr{\circ}{\mathbb P}}(\gamma_{1,n}^{-k,o},\cdot)$ at $\gamma_{1,n}^{k,o}\in {\cal U}_{1,n}^k$   in the direction $\gamma_{1,n}^{k} -\gamma_{1,n}^{k,o}\in {\cal U}_{1,n}^k$  is computed from
\begin{equation*}
     \lim_{ \varepsilon \downarrow 0}   \frac{J^{\sr{\circ}{\mathbb P}}(\gamma_{1,n}^{-k,o}, \gamma_{1,n}^{k,\veps})-J^{\sr{\circ}{\mathbb P}}(\gamma_{1,n}^{-k,o}, \gamma_{1,n}^{k,o})}{\varepsilon} \equiv \frac{d}{d \varepsilon} J^{\sr{\circ}{\mathbb P}}(\gamma_{1,n}^{-k,o}, \gamma_{1,n}^{k,\veps})\Big\vert_{\varepsilon =0}
\end{equation*}
On measure $\sr{\circ}{\mathbb P}$,  $(x_{1,n}, y_{1,n}^{(M)})$ do not depend on $u^{(K)}$, hence  $\frac{d}{d \varepsilon}J^{\sr{\circ}{\mathbb P}}(\gamma_{1,n}^{-k,o}, \gamma_{1,n}^{k,\veps})\Big\vert_{\varepsilon =0}=\mbox{the right side of (\ref{vaeqst_b})}$. 
Writing (\ref{vaeqst_b}) component-wise we obtain (\ref{vaeqst_c}).
From (\ref{vaeqst_c}), for each $k$,  letting $\gamma_s^k=\gamma_s^{k,o}, \forall s \neq t$, and   reconditioning on ${\cal F}^{I_t^k}$  we have, 
\begin{align}
&{\bf  E}^{\sr{\circ}{\mathbb P}} \Big\{{\bf  E}^{\sr{\circ}{\mathbb P}} \Big\{  \nabla_{u_t^k} L(x_{1,n}, y_{1,n}^{(M)},\gamma_{1,n}^{-k,o}, \gamma_{1,t-1}^{o},u_{t}^{k},\gamma_{t+1,n}^{o})\big|_{u_t^k=\gamma_t^{k,o}}\nonumber \\
&. \big( \gamma_t^k-\gamma_t^{k,o}\big)\Big| {\cal F}^{I_t^k} \Big\}  \Big\}\geq 0, \;  \forall  \gamma_{t}^k  \in {\cal  U}_{t}^{k}, \;  \;  \forall (k,t).  \label{vaeqst_con_a}
 \end{align}
Since $\gamma_t^k-\gamma_t^{k,o}\in {\cal  U}_{t}^{k}$ is  ${\cal F}^{I_t^k}-$measurable we obtain (\ref{vaeqst_con}). The statements under  (1.2) follow from (1.1) using the inverse change of probability measure   $ d\sr{\circ}{\mathbb  P}=\big(\Theta_n^u\big)^{-1} \Big|_{{\cal F}_n} d{\mathbb  P}^{u}, \forall n$. In particular,   (\ref{vaeqst_b_orig}) follows from  (\ref{vaeqst_b})  by the inverse change of measure and (\ref{vaeqst_c_orig}) follows from  (\ref{vaeqst_c})  by using the martingale property of $\big(\Theta_n^u\big)^{-1}$ similar to the derivation leading to (\ref{i8-or-4}).  \\
(2) To show the stationary conditions of PbP optimality  (\ref{vaeqst_con}) imply global optimality, we make use of  convexity (A5), i.e., 
 we have for vectors $u_{1,n}^{k}\in {\mathbb A}_{1,n}^k$, 
\begin{align}
&L(x_{1,n}, y_{1,n}^{(M)},u_{1,n}^{(K)}) -  L(x_{1,n}, y_{1,n}^{(M)}, u_{1,n}^{(K),o})\nonumber \\
\geq  &\sum_{k=1}^K \nabla_{u_{1,n}^k}L(x_{1,n}, y_{1,n}^{(M)},u_{1,n}^{-k,o},u_{1,n}^{k})\big|_{u_{1,n}^{k}=u_{1,n}^{k,o}}\nonumber \\
&. \Big(u_{1,n}^{k} -u_{1,n}^{k,o}\Big), \; \forall u_{1,n}^{(K),o}\in {\mathbb A}_{1,n}^{(K)}, \;  \forall u_{1,n}^{(K)}\in {\mathbb A}_{1,n}^{(K)}. \label{suff-p2}
\end{align}
Then 
\begin{align}
&  J^{\sr{\circ}{\mathbb P}}(\gamma_{1,n}^{1, o}, \ldots, \gamma_{1,n}^{K,o})-  J^{\sr{\circ}{\mathbb P}}(\gamma_{1,n}^{1},\ldots,  \gamma_{1,n}^{K}) \\
  &\leq  -{\bf  E}^{\sr{\circ}{\mathbb P}}\Big\{  \sum_{k=1}^K \nabla_{u_{1,n}^k}L(x_{1,n}, y_{1,n}^{(M)},\gamma_{1,n}^{-k,o},u_{1,n}^{k})\big|_{u_{1,n}^{k}=\gamma_{1,n}^{k,o}}\nonumber \\
&. \Big(\gamma_{1,n}^{k} -\gamma_{1,n}^{k,o}\Big) \Big\},   \hso  \forall (\gamma_{1,n}^{1}, \ldots, \gamma_{1,n}^{K})\in {\cal U}_{1,n}^{(K)} \label{suff-p3}\\
 &=-\sum_{k=1}^K \sum_{t=1}^n{\bf  E}^{\sr{\circ}{\mathbb P}} \Big\{\nonumber \\
 &  \nabla_{u_t^k} L(x_{1,n}, y_{1,n}^{(M)},\gamma_{1,n}^{-k,o}, \gamma_{1,t-1}^{o}, u_{t}^{k},\gamma_{t+1,n}^{o})\big|_{u_t^k=\gamma_t^{k,o}}\nonumber \\
&. \big( \gamma_t^k-\gamma_t^{k,o}\big) \Big\}
\end{align}
\begin{align}
&=-  \sum_{k=1}^K \sum_{t=1}^n  {\bf  E}^{\sr{\circ}{\mathbb P}}\Big\{  \nonumber \\
& {\bf  E}^{\sr{\circ}{\mathbb P}}\Big\{
\nabla_{u_t^k} L(x_{1,n}, y_{1,n}^{(M)},\gamma_{1,n}^{-k,o}, \gamma_{1,t-1}^{o}, u_{t}^{k},\gamma_{t+1,n}^{o})\big|_{u_t^k=\gamma_t^{k,o}}\nonumber \\
&. \big( \gamma_t^k-\gamma_t^{k,o}\big) 
\Big|  {\cal F}^{I_t^k} \Big\} \Big\},   \;  \forall (\gamma_{1,n}^{1}, \ldots, \gamma_{1,n}^{K})\in {\cal U}_{1,n}^{(K)} \label{suff-p4}\\
&\leq 0,\hso  \forall (\gamma_{1,n}^{1}, \ldots, \gamma_{1,n}^{K})\in {\cal U}_{1,n}^{(K)} \hso \mbox{if   (\ref{vaeqst_con}) holds}. \label{suff-p5}
\end{align}
By (\ref{suff-p5}) then PbP optimality (\ref{vaeqst_con}) imply  global optimality. 
\end{proof}
 
 \begin{remark} (Some Generalizations) \\
 Theorem~\ref{st-ksm}, although, general,  it can  be  modified to cover alternative   conditional  PMs, such as, 
 \begin{align*}
S_{t+1}(dx_{t+1}|x_{1,t}, y_{1,t}^{(M)},  u_{1,t}^{(K)}),\;  Q_{t+1}^{(M)}(dy_{t+1}^{(M)} |x_{1,t}, y_{1,t}^{(M)}, u_{1,t}^{(K)})
\end{align*}
and alternative   recursive models to Example~\ref{ex:rec},  such as, 
\begin{align}
&x_{t+1}=f(t,x_{1,t}, y_{1,t}^{(M)},  u_{1,t}^{(M)}, w_{t+1}), \;  t=0, \ldots,  n-1, \label{NDM-3-g}\\
& y_{t+1}^{(M)}   = h^{(M)}(t,x_{1,t},y_{1,t}^{(M)},  u_{1,t}^{(K)}, v_{t+1}^{(M)}), \label{NDM-4-g} 
\end{align}
where the RV $(x_0,y_0^{(M)})$ is independent of RVs $(w_{1,n}, v_{1,n}^{(M)})$, and their PMs are fixed,  ${\bf P}_{x_{0}, y_0^{(M)}}(dx_{0}, dy_0^{(M)})$ and 
\begin{align*}
 {\bf P}_{w_{1,n}, v_{1,n}^{(M)}}(dw_{1,n},dv_{1,n}^{(M)} )= \prod_{t=1}^{n}  \Psi_t(dw_t)\Phi_t^{(M)}(dv_t^{(M)}).
\end{align*} 
\end{remark}

\section{PbP  Strategies of the Counterexample}
\label{sec::optimalstrategy}
In this section, we consider the counterexample \cite{WitsenhausenOriginal}, and  we invoke the change of measure of Section~\ref{discrete}, to  derive the two  optimal strategies (\ref{nl_1})-(\ref{eq:gamma2_in}), and to  prove a fixed point theorem.
 for existence and uniqueness of solutions to the integral equations.

\subsection{The Witsenhausen Counterexample and Related Literature}
\label{sect:count}

{\it Statement of \cite{WitsenhausenOriginal}.} Consider a probability space $\Big(\Omega, {\cal F}, {\mathbb P}^u\Big), u \tri (u_1, u_2)$ and  two independent random variables (RVs), defined on it with finite second moments\footnote{Witsenhausen \cite{WitsenhausenOriginal} considered the value $\sigma^2=1$.}
\begin{align}
{\mathbb P}^u: \left\{ \begin{array}{l} (x_0,v) : \Omega \rar {\mathbb R}^2,  \; {\bf  E}^{{\mathbb P}^u} (x_0)^2= \sigma_{x}^2 <\infty, \\
{{\bf  E}^{{\mathbb P}^u} (v)^2}<\infty,\;
\mbox{arbitrary  ${\bf P}_{x_0}$,  ${\bf P}_{v}$ with}, \\
{\bf  E}^{{\mathbb P}^u}\{x_0\} = {\bf  E}^{{\mathbb P}^u}\{v\}=0, \;  {\bf E}^{{\mathbb P}^u}(v)^2=\sigma^2.
\end{array} \right. 
 \label{dist}
\end{align}
The  stochastic optimal control problem is described below. 
\begin{align}
&\mbox{\it  State Equations.} \hso   x_1 = x_0 + u_1,\hst  x_2 = x_1 - u_2. \label{eq:stateeq}\\
&\mbox{\it Output Equations.}\hso 
y_0 = x_0,  \hst 
y_1 = x_1 +v.  \label{eq:opeq}\\
&\mbox{\it Aver. Payoff.} \hso 
J^{{\mathbb P}^u}(\gamma_1,\gamma_2) =   {\bf  E}^{{\mathbb P}^u}\Big\{ k^2 (u_1)^2 + (x_2)^2 \Big\}.    \label{eq:costeq}\\
&\mbox{{\it Admissible Strategies.} A tuple of  Borel measurable fuct.}\nonumber \\
&  \gamma\tri (\gamma_1, \gamma_2)\in {\cal A}_{ad}, \hso   
u_1 = \gamma_1(y_0),  \hst 
u_2 = \gamma_2(y_1) . \label{eq:admissiblestrategies}
\end{align}
%
%
Here, $ k^2 >0$ and  we take ${\cal F}\tri \sigma\{x_0,x_1,x_2,y_1, u_1, u_2\}$. 

\noi{\it Objective.} Given the distributions ${\bf P}_{x_0}(dx_0), {\bf P}_{v}(dv)$, i.e., ${\bf P}_{x_0,v}(dx_0,dv)={\bf P}_{x_0}(dx_0) {\bf P}_{v}(dv)$ such that (\ref{dist})-(\ref{eq:admissiblestrategies}) hold,  minimize over $ (\gamma_1, \gamma_2) \in {\cal A}_{ad}$ the average payoff,
\begin{align}
{\mathbb P}^{\gamma}: \hso J^{{\mathbb P}^{\gamma^o}}(\gamma^o) \tri  \inf_{\gamma\in {\cal A}_{ad}}  J^{{\mathbb P}^{\gamma}}(\gamma), \hso  \gamma \tri (\gamma_1, \gamma_2) .
 \label{eq:costeq_new}
\end{align}

\noi{\it Restatement of the Counterexample \cite{WitsenhausenOriginal}.} Witsenhausen considered  the equivalent re-formulation of problem (\ref{eq:costeq_new}) given by
\begin{align}
{\mathbb P}^{\gamma}:  \left\{ \begin{array}{l}
\overline{\gamma}_1(x_0)\tri x_0 + \gamma_1(x_0),  \hso y_1=\overline{\gamma}_1(x_0)+v, \\
 x_1= \overline{\gamma}_1(x_0), \; x_2=\overline{\gamma}_1(x_0)-\gamma_2(\overline{\gamma}_1(x_0)+ v), \\
J^{{\mathbb P}^{\gamma}}(\overline{\gamma}_1^o,\gamma_2^o) \tri \inf_{(\overline{\gamma}_1, \gamma_2)\in {\cal A}_{ad}} J^{{\mathbb P}^{\gamma}}(\overline{\gamma}_1,\gamma_2)  , \\
J^{{\mathbb P}^{\gamma}}(\overline{\gamma}_1,\gamma_2) \tri    {\bf E}^{{\mathbb P}^\gamma}\Big\{ k^2 \Big(x_0-\overline{\gamma}_1(x_0)\Big)^2 \\
+ \Big(\overline{\gamma}_1(x_0)-\gamma_2(\overline{\gamma}_1(x_0)+ v)\Big)^2 \Big\} \equiv  J^{{\mathbb P}^{\gamma}}(\gamma_1,\gamma_2).
\end{array} \right. 
 \label{res_1}
\end{align}
A more  general problem (not addressed in \cite{WitsenhausenOriginal}) is the following. 

\ \
\noi {\it Problem G,   $\pi^{{\mathbb P}^{\gamma}}(k^2,{\bf P}_{x_0},{\bf P}_{v})$}. The general problem,  $\pi^{{\mathbb P}}(k^2,{\bf P}_{x_0}, {\bf P}_{v})$, is to minimize $J^{{\mathbb P}^{\gamma}}({\gamma}_1,\gamma_2)$ over ${\cal A}_{ad}$  or equivalently $J^{{\mathbb P}^{\gamma}}(\overline{\gamma}_1,\gamma_2)$, subject to (\ref{dist})  with $(v, x_0)$ having arbitrary distributions, not necessarily Gaussian. \\

\noi The  problems, which are  investigated in \cite{WitsenhausenOriginal} are,   Problem \#1 with $v$ a Gaussian RV,   and Problem \#2 with $v$ and $x_0$ both  Gaussian RVs, as defined below.  

 {\it Problem \#1,   $\pi^{{\mathbb P}^{\gamma}}(k^2,{\bf P}_{x_0}, G(0,\sigma^2))$ of  \cite{WitsenhausenOriginal}}. The  problem,  $\pi^{{\mathbb P}^{\gamma}}(k^2,{\bf P}_{x_0}, G(0,\sigma^2))$, is to minimize $J^{{\mathbb P}^{\gamma}}({\gamma}_1,\gamma_2)$ over ${\cal A}_{ad}$  or  $J^{{\mathbb P}}(\overline{\gamma}_1,\gamma_2)$,  subject to (\ref{dist}) with $v$ a Gaussian RV with mean zero and variance $\sigma^2>0$, i.e., $v \in G(0,\sigma^2)$. In  \cite{WitsenhausenOriginal}, $\sigma^2=1$.  

{\it Problem  \#2, $\pi^{{\mathbb P}^{\gamma}}(k^2,G(0,\sigma_x^2), G(0,\sigma^2))$ of  \cite{WitsenhausenOriginal}}. The Gaussian problem, $\pi^{{\mathbb P}^{\gamma}}(k^2,G(0,\sigma_x^2),G(0,\sigma^2) )$  is to minimize $J^{{\mathbb P}^\gamma}({\gamma}_1,\gamma_2)$ over ${\cal A}_{ad}$  or $J^{{\mathbb P}^\gamma}(\overline{\gamma}_1,\gamma_2)$,  subject to (\ref{dist}) with $x_0\in G(0,\sigma_x^2), \sigma_x^2>0 $ and $v \in G(0,\sigma^2), \sigma^2>0$.  


\ \

For  Problem \#1,  $\pi^{{\mathbb P}^{\gamma}}(k^2,{\bf P}_{x_0}, G(0,1))$, and Problem \#2, $\pi^{{\mathbb P}^{\gamma}}(k^2,G(0,\sigma_x^2), G(0,1))$,  Witsenhausen \cite{WitsenhausenOriginal} derived the following  properties (re-confirmed some from our results). 

1)  Problem \#1,   $\pi^{{\mathbb P}^{\gamma}}(k^2,{\bf P}_{x_0}, G(0,1))$.  \\
1.1) An  optimal strategy $(\overline{\gamma}_1^o, \gamma_2^o)\in {\cal A}_{ad}$ exists  and satisfies $0 \leq J^{{\mathbb P}^{\gamma}}(\overline{\gamma}_1^o,\gamma_2^o) \leq \min\{1, k^2 \sigma_x^2\}$   (Theorem~1 in  \cite{WitsenhausenOriginal}).  \\
1.2)  If  ${\bf  E}^{{\mathbb P}^{\gamma}}\big(\overline{\gamma}_1(x_0)\big)^2<\infty$ and $\overline{\gamma}_1$ is fixed,  then $\gamma_2^o(y_1)={\bf E}^{{\mathbb P}^{\gamma}}\big\{\overline{\gamma}_1(x_0)\big|y_1\big\}$ (Lemma~3.(c) in  \cite{WitsenhausenOriginal}). \\
1.3) If    ${\bf P}_{x_0}$ is  restricted to a two-point symmetric  distribution with mass of $\frac{1}{2}$ at $x_0=\sigma_x>0$ and  $\frac{1}{2}$ at $x_0=-\sigma_x$, then the optimal strategies are $\overline{\gamma}_1(x_0)=\frac{a}{\sigma_x}x_0, \gamma_2(y_1)=a \tanh (a y_1)$ for some $a$ that satisfies a certain equation (Lemma~15 in  \cite{WitsenhausenOriginal}).

2)  Problem  \#2, $\pi^{{\mathbb P}^{\gamma}}(k^2,G(0,\sigma_x^2), G(0,1))$. \\ 
2.1) If $(\overline{\gamma}_1(\cdot),\gamma_2(\cdot))$ are restricted   to affine (linear) strategies with corresponding optimal payoff defined by 
\begin{align*}
J^{{\mathbb P}^{\gamma},wa} \tri  \inf \Big\{ J^{{\mathbb P}^{\gamma}}(\overline{\gamma}_1^{wa},\gamma_2^{wa}) \Big| (\overline{\gamma}_1^{wa},\gamma_2^{wa})=(\lambda\; x_0, \mu \; y_1)\Big\}
\end{align*}
for $(\mu, \lambda)\in {\mathbb R}^2$, then the optimal  strategies are  (Sect. 4 in \cite{WitsenhausenOriginal})
\begin{align}
 &   \overline{\gamma}_1^{wa}(x_0) = \lambda \;  x_0, , \hso 
    \gamma_2^{wa}(y_1) = \mu \; y_1, \;     \mu = \frac{\sigma_x^2 \lambda^2}{1+\sigma_x^2 \lambda^2}\\
& t=\sigma_x \lambda \; \mbox{a real root of} \; 
    (t-\sigma_x)(1+t^2)^2 + \frac{1}{k^2}t = 0.
\end{align}
2.2)  There exist parameter values $(k^2, \sigma_x^2)$ such that the optimal payoff $J^{{\mathbb P}^{\gamma^o}}(\overline{\gamma}_1^o, \gamma_2^o)$ is less than the optimal payoff when $(\overline{\gamma}_1(\cdot),\gamma_2(\cdot))$ are restricted to  affine strategies (Theorem~2 in  \cite{WitsenhausenOriginal}). In particular,  there exist parameters $(k,\sigma_x^2)\in (0,\infty)\times (0,\infty)$ such that the tuple of nonlinear sub-optimal strategies
\begin{align}
  &  \overline{\gamma}_1^{wn}(x_0)= \sigma_x \; \text{sgn}(x_0), \hso  
    \gamma_2^{wn}(y_1)= \sigma_x \tanh{(\sigma \;y_1)}, 
    \label{eq:WitsenhausenNonlinearLaws}
\end{align}
incur a payoff $J^{{\mathbb P}^{\gamma}, wn}\tri J^{{\mathbb P}^{\gamma}}(\overline{\gamma}_1^{wn},\gamma_2^{wn})$ which is smaller than the optimal payoff $J^{{\mathbb P}^{\gamma},wa}$ incur by all affine strategies (Theorem~2 1 in \cite{WitsenhausenOriginal}). 
That is, there exist parameters $(k,\sigma_x^2)$ such that $J^{{\mathbb P}^{\gamma}, wn} < J^{{\mathbb P}^{\gamma}, wa}$. It is also shown that  $J^{{\mathbb P}^{\gamma}, wn} < J^{{\mathbb P}^{\gamma}, wa}$ as $k \to 0$. However, this does not mean nonlinear strategies outperform affine strategies for all  $(k,\sigma_x^2)$. 

\ \

3)  Problem \# 1,  $\pi^{{\mathbb P}^{\gamma}}(k^2,F_{x_0}, G(0,1))$ and Problem \# 2, $\pi^{{\mathbb P}^{\gamma}}(k^2,G(0,\sigma_x^2), G(0,1))$ with emphasis on the latter received  immense attention
in the literature i.e., 
 \cite{BansalWhenisAffineOptimal,neuralnetworksolution,iterativecodingWit,hierarchialLee,McEany2,Grover2013,Submarmanian_newresults}. 
 Prior studies provide numerical techniques to solve Problems $\pi^{{\mathbb P}^{\gamma}}(k^2,F_{x_0}, G(0,1))$ and  $\pi^{{\mathbb P}}(k^2,G(0,\sigma_x^2), G(0,1))$, by using  properties of optimal strategies derived by Witsenhausen such as, $\gamma_2^o(y_1)={\mathbb E}^{{\mathbb P}^{\gamma_1, \gamma_2^o}}\big\{\overline{\gamma}_1(x_0)\big|y_1\big\}$ (Lemma~3.(c) in  \cite{WitsenhausenOriginal}) and  (\ref{eq:WitsenhausenNonlinearLaws}). 

\subsection{ Equivalent Optimization Problem on  the Reference Probability Space $\big(\Omega, {\cal F}, \sr{\circ}{\mathbb P}\big)$}
Theorem~\ref{thm:eq-c} follows from  Theorem~\ref{thm:RND_1} and Theorem~\ref{thm:p-off}.

%
%

\begin{theorem} (The Equivalent counterexample problems)\\
\label{thm:eq-c}
The original counterexample Problem G,   $\pi^{{\mathbb P}^{\gamma}}(k^2,{\bf P}_{x_0}, {\bf P}_{v_0})$, defined on  probability space $\big(\Omega, {\cal F}, {\mathbb P}^{\gamma}\big)$ (i.e.,  $J^{{\mathbb P}^{\gamma}}({\gamma}_1^o, \gamma_2^o)=J^{{\mathbb P}^{\gamma}}(\overline{\gamma}_1^o, \gamma_2^o)=\mbox{(\ref{eq:costeq_new}) $=$(\ref{res_1}) }$ subject to  (\ref{dist})-(\ref{eq:admissiblestrategies})) is equivalent to Problem G-Eqv,    $\pi^{\sr{\circ}{\mathbb P}}(k^2,{\bf P}_{x_0}, {\bf P}_{v_0})$ defined under the reference probability space $\big(\Omega, {\cal F}, \sr{\circ}{\mathbb P}\big)$, as stated below. 
\begin{align}
&\mbox{\it  Problem G-Eqv,    $\pi^{\sr{\circ}{\mathbb P}}(k^2,{\bf P}_{x_0}, {\bf P} _{v_0})$.} \label{prob_3}  \\
&\Big({\cal F}, \sr{\circ}{\mathbb P}\Big):  \left\{ \begin{array}{l}  J^{\sr{\circ}{\mathbb P}}(\overline{\gamma}_1^o,\gamma_2^o) =\inf_{(\overline{\gamma}_1, \gamma_2)\in {\cal A}_{ad}} 
{\bf E}^{\sr{\circ}{\mathbb P}}\Big\{  \Lambda^{\overline{\gamma}_1, \gamma_2}(x_0,y_1) \\
 .\Big[k^2 \Big(x_0-\overline{\gamma}_1(x_0)\Big)^2 + \Big(\overline{\gamma}_1(x_0)-\gamma_2(y_1)\Big)^2\Big] \Big\}, \\
  x_1=\overline{\gamma}_1(x_0), \: x_2= \overline{\gamma}_1(x_0)-\gamma_2(y_1), \; y_1=v,\\ 
\mbox{such that  the folowing hold: }
\end{array} \right. \nonumber \\
&(i)\; \mbox{$\Lambda^{\overline{\gamma}_1, \gamma_2}(x_0,y_1) $ is the Radon-Nikodym Derivative of the} \nonumber \\
&\mbox{original measure ${\mathbb P}^{\gamma}$ w.r.t.  reference measure $\sr{\circ}{\mathbb P}$ defined by  }    \nonumber \\
& {\mathbb P}^{\gamma}(dx_0, dx_1, dy_1)=\Lambda^{\overline{\gamma}_1, \gamma_2}(x_0,y_1)  \sr{\circ}{\mathbb P}(dx_0, dx_1, dy_1),  \label{RND_1w}  \\
&\Lambda^{\overline{\gamma}_1, \gamma_2}(x_0,y_1)\tri  \frac{Q^{\overline{\gamma}_1, \gamma_2}(dy_1|x_1, x_0)}{{\bf P}_{v}(dy_1)}, \hso  x_1=\overline{\gamma}_1(x_0),  \label{RND_2w} \\
&Q^{\overline{\gamma}_1, \gamma_2}(dy_1|x_1, x_0)
= {\bf P}_v\big(v: \overline{\gamma}_1(x_0)+v \in dy_1\big).\label{RND_3w} \\
&(ii)\;  \mbox{Under the reference measure $\sr{\circ}{\mathbb P}$, the RVs    $(x_0,y_1)$},\nonumber  \\
&\mbox{ are independent   with  distributions   $({\bf P}_{x_0}, {\bf P}_{y_1}={\bf P}_{v})$}.\nonumber  \\
&(iii) \; \mbox{The expectation  ${\bf E}^{ \sr{\circ}{\mathbb P}}\{\cdot\}$ is w.r.t.  the measure }\nonumber\\
&\mbox{ $\sr{\circ}{\mathbb P}(dx_0, dy_1)={\bf P}_{x_0, v}(dx_0, dy_1)= {\bf P}_{x_0}(dx_0) {\bf P}_{v}(dy_1)$}. 
\end{align}
\begin{align}
&\mbox{(iv) If the probability density functions exist then}\nonumber  \\
&Q^{\overline{\gamma}_1, \gamma_2}(dy_1|x_1, x_0)= f_{v}(y_1-\overline{\gamma}_1(x_0))dy_1, \; \\
&{\bf P}_{v}(dy_1)= f_{v}(y_1)dy_1, \\
&\Lambda^{\overline{\gamma}_1, \gamma_2}(x_0,y_1)=\frac{f_{v}(y_1-\overline{\gamma}_1(x_0))}{f_{v}(y_1)}   .   \label{RND_4w}
\end{align} 
\end{theorem}
\begin{proof} 
Due to Theorem~\ref{thm:RND_1}, and Theorem~\ref{thm:p-off}, by using  ${\mathbb  P}^u(B)={\mathbb  P}(B)=\int_B   \Lambda_t^u(\omega)  d\sr{\circ}{ {\mathbb P}}(\omega)$, $\Lambda_t^u(\omega)=\Lambda^{\overline{\gamma}_1, \gamma_2}(x_0,y_1),    \forall 
B\in {\cal  F}$, i.e., we do not use $M_t^u$ to change the measure of $(x_0,x_1,x_2)$.  
\end{proof}

\subsection{Optimal Strategies of the Counterexample }
\label{sec:opt}
In Theorem~\ref{thm;w-g}, we determine the equations satisfied by the  optimal strategy $(\overline{\gamma}_1^o,\gamma_2^o)=(x_0+{\gamma}_1^o,\gamma_2^o)$ using the equivalent  
 Problem G-Eqv,    $\pi^{\sr{\circ}{\mathbb P}}(k^2,{\bf P}_{x_0}, {\bf P}_{v_0})$.

\begin{theorem} (Stationary conditions-Problem G-Eqv)\\
\label{thm;w-g}
Consider   Problem G-Eqv,    $\pi^{\sr{\circ}{\mathbb P}}(k^2,{\bf P}_{x_0}, {\bf P}_{v_0})$, of Theorem~\ref{thm:eq-c}, 
 and assume the following. \\
 (a.i) The density of the RND is 
\begin{align*}
\Lambda^{u_1, u_2}(x_0,y_1)\tri  \frac{Q^{u_1, u_2}(dy_1|x_1, x_0)}{{\bf P}_{v}(dy_1)}=\frac{f_{v}(y_1-x_0-u_1)}{f_{v}(y_1)}
\end{align*}
i.e., the probability density functions exist, and $f_{v}(y_1)>0, f_{v}(y_1-x_0-u_1)>0, \forall (x_0,u_1, y_1)$. \\
(a.ii) $ \Lambda^{{u}_1,u_2}(x_0,y_1) $ is continuously differentiable in  $({u}_1, u_2)$ uniformly over $(x_0,y_1)$ and the  derivative is an  element of $L^2$. \\
(a.iii) The Gateaux differential of  $J^{\sr{\circ}{\mathbb P}}(\cdot,\cdot): L^2 \times L^2 \rar [0,\infty)$ at  $({\gamma}_1^o, \gamma_2^o)$ in the direction      $({\gamma}_1, \gamma_2) -({\gamma}_1^o, \gamma_2^o   )\in L^2 \times L^2$ exists. \\
Define, 
\begin{align*}
L^{u_1,u_2}(x_0, y_1) \tri  \Lambda^{u_1, u_2}(x_0, y_1)  \Big(k^2  u_1^2 + (x_0+u_1-u_2)^2\Big).
\end{align*}
and introduce  the derivatives, 
\begin{align}
&\nabla_{u_1} L^{u_1,u_2}(x_0, y_1) =2\Lambda^{u_1,u_2}(x_0, y_1) \Big( k^2  u_1 + (x_0+u_1-u_2)\Big)\nonumber \\
&+\nabla_{u_1}\Lambda^{u_1,u_2}(x_0, y_1) \Big(k^2  u_1^2 + (x_0+u_1-u_2)^2\Big), \\
&\nabla_{u_2}L^{u_1,u_2}(x_0, y_1) =-2\Lambda^{u_1,u_2}(x_0, y_1) (x_0+u_1-u_2)
\end{align}
The following hold. \\
(i) The two conditional stationary condition  under probability measure $\sr{\circ}{\mathbb P}$ hold, 
\begin{flalign}
 &  {\mathbb E}^{\sr{\circ}{\mathbb P}} \Big\{    \nabla_{u_1}L^{{u}_1, \gamma_2^{o}}(x_0, y_1)\big\vert_{{u}_1={\gamma}_1^o(x_0)}  
         \Big({\gamma}_1(x_0)-{\gamma}_1^o(x_0)\Big) \Big| x_0 \Big\} \nonumber  
         \\  & \hst     \geq 0, \hso \sr{\circ}{\mathbb P}\big|_{x_0}-a.s., \hso  \forall {\gamma}_1\in L^2 , 
         \label{w10-pbp-1-p_a}\\
  &    {\mathbb E}^{\sr{\circ}{\mathbb P}} \Big\{ \nabla_{u_2}L^{{\gamma}_1^o, u_2}(x_0, y_1)\big\vert_{u_2=\gamma_2^o(y_1)} 
         \Big(\gamma_2(y_1)-\gamma_2^o(y_1)\Big) \Big| y_1 \Big\} \nonumber 
         \\  
         & \hst \geq 0, \hso   \sr{\circ}{\mathbb P}\big|_{y_1}-a.s., \hso \forall \gamma_2\in L^2 .      \label{w10-pbp-1-p_b}
\end{flalign}
(ii) The two conditional stationary condition  under probability measure ${\mathbb P}^{\gamma}$ hold, 
\begin{flalign}
 &  {\mathbb E}^{{\mathbb P}^{\gamma^o}} \Big\{ \Big(\Lambda^{\gamma_1^o, \gamma_2^o}(x_0,y_1) \Big)^{-1}   \nabla_{u_1}L^{{u}_1, \gamma_2^{o}}(x_0, y_1)\big\vert_{{u}_1={\gamma}_1^o(x_0)}  \nonumber  
         \\  &
      .   \Big({\gamma}_1(x_0)-{\gamma}_1^o(x_0)\Big) \Big| x_0 \Big\}   \geq 0, \; {\mathbb P}^{\gamma^o}\big|_{x_0}-a.s., \;  \forall {\gamma}_1\in L^2 , 
         \label{w10-pbp-1-o_a}\\
  &    {\mathbb E}^{{\mathbb P}^{\gamma^o}} \Big\{ \Big(\Lambda^{\gamma_1^o, \gamma_2^o}(x_0,y_1) \Big)^{-1} \nabla_{u_2}L^{{\gamma}_1^o, u_2}(x_0, y_1)\big\vert_{u_2=\gamma_2^o(y_1)}  \nonumber 
         \\  
         & 
  .       \Big(\gamma_2(y_1)-\gamma_2^o(y_1)\Big) \Big| y_1 \Big\}\geq 0,   {\mathbb P}^{{\gamma^o}}\big|_{y_1}-a.s., \forall \gamma_2\in L^2 .      \label{w10-pbp-1-o_b}
\end{flalign}
\end{theorem}
\begin{proof} The statements follow directly from Theorem~\ref{st-ksm} applied to  Problem G-Eqv,    $\pi^{\sr{\circ}{\mathbb P}}(k^2,{\bf P}_{x_0}, {\bf P}_{v_0})$,    (\ref{prob_3})-(\ref{RND_4w}). 
\end{proof}

\ \

In Theorem~\ref{thm;w-sp-1}, we determine the exact equations satisfied by the  optimal strategy $\gamma^o= ({\gamma}_1^o,\gamma_2^o)$ for 
 Problem \# 1.

\begin{theorem} (Optimal strategies for  Problem \# 1,    $\pi^{{\mathbb P}^{\gamma}}(k^2,{\bf P}_{x_0}, G(0,\sigma^2))$)\\
\label{thm;w-sp-1}
Consider the statement of Theorem~\ref{thm;w-g}, for the special case of  Problem \# 1,    $\pi^{{\mathbb P}^{\gamma}}(k^2,{\bf P}_{x_0}, G(0,\sigma^2))$.
The following hold. \\  
 (i) The density of the RND is given by 
\begin{align*}
\Lambda^{u_1, u_2}(x_0,y_1)
= \frac{\exp\left\{-\frac{(y_1 -x_0-u_1)^2}{2\sigma^2}\right\}}{\exp\left\{-\frac{y_1^2}{2\sigma^2}\right\}}>0, \hso \forall (x_0, u_1, y_1). 
\end{align*}
(ii) The optimal  strategies $({\gamma}_1^o, \gamma_2^o)  \in {\cal A}_{ad}$ exist.\\
(iii)  The   optimal control strategies $(\gamma_1^o, \gamma_2^o)$ on the  original probability measure ${\mathbb P}^{\gamma}$ satisfy the equations, 
\begin{align}
&{\bf  E}^{{\mathbb P}^{\gamma^o}}\Big\{\Big[ \frac{y_1-x_0-u_1}{\sigma^2} \Big(k^2  u_1^2 + (x_0+u_1-\gamma_2^o(y_1))^2  \Big) + 2k^2u_1 \nonumber \\
& +2(x_0+u_1-\gamma_2^o(y_1)) \Big]   \Big \vert x_0 \Big\}   \Big\vert_{u_1=\gamma_1^o(x_0)}=0,   {\mathbb P}^{\gamma^o}\big\vert_{ x_0}-a.s,  \label{ws16} \\ 
&{\bf  E}^{{\mathbb P}^{\gamma^o}}\Big\{x_0+\gamma_1^o(y_0)-u_2\Big\vert y_1 \Big\}\Big\vert_{u_2=\gamma_2^o(y_1)}=0,   {\mathbb P}^{\gamma^o}\big\vert_{ y_1}-a.s.  \label{ws17}
\end{align}
where $y_1=x_0+{\gamma}_1^o(x_0)+v$. 
Moreover, 
\begin{align}
{\gamma}_1^o(x_0) =&- \frac{1}{k^2}{\bf E}^{{\mathbb P}^{\gamma^o}}\Big\{x_0+{\gamma}_1^o(x_0) - \gamma_2^o(y_1) \Big|x_0 \Big\}\nonumber \\
&
-\frac{1}{2k^2\sigma^2} {\bf E}^{{\mathbb P}^{\gamma^o}}\Big\{   \Big(y_1 -x_0  -{\gamma}_1^o(x_0)\Big)\nonumber \\
&. \Big(x_0+{\gamma}_1^o(x_0)  - \gamma_2^o(y_1)\Big)^2\Big|x_0 \Big\},
 \label{nl_1_th}   \\
\gamma_2^o(y_1)=&   {\bf E} ^{{\mathbb   P}^{\gamma^o}}\Big\{ x_0+{\gamma}_1^o(x_0)    \Big| y_1 \Big\}. \label{nl_2_th}
\end{align}
Equivalently, the optimal control strategies $(\overline{\gamma}_1^o, \gamma_2^o)=(x_0+\gamma_1^o, \gamma_2^o)$ satisfy (\ref{nl_1}),   (\ref{nl_2})  and 
 the two integral equations  (\ref{eq:gamma1bar_in}),   (\ref{eq:gamma2_in}).
\end{theorem}
\begin{proof} (i) Follows from  Theorem~\ref{thm;w-g}, with ${\bf P}_{v_0}= G(0,\sigma^2)$.  (ii) This is due to
Witsenhausen \cite{WitsenhausenOriginal} (Theorem 1). 
 (iii) By Theorem~\ref{thm;w-g}.(iii),  (\ref{w10-pbp-1-p_a}),  and the existence of the optimal  strategies $({\gamma}_1^o, \gamma_2^o)  \in {\cal A}_{ad}$, we must have that inequality holds with equality, to deduce, 
\begin{flalign}
&{\bf  E}^{\sr{\circ}{\mathbb P}}\Big\{ \Lambda^{u_1, \gamma_2^o}(x_0, y_1)\Big[ \frac{y_1-x_0-u_1}{\sigma^2} \Big(k^2  u_1^2 +   \label{ws14}\\
&(x_0+u_1-\gamma_2^o(y_1))^2  \Big)+ 2k^2u_1+ \nonumber \\
&2(x_0+u_1-\gamma_2^o(y_1)) \Big]   \Big \vert x_0 \Big\}   \Big\vert_{u_1=\gamma_1^o(x_0)}=0, \hso  \sr{\circ}{\mathbb P}\big\vert_{ x_0}-a.s
\nonumber
\end{flalign}
Similarly, by Theorem~\ref{thm;w-g}.(iii),  (\ref{w10-pbp-1-p_b}), we have
\begin{flalign}
{\bf E}^{\sr{\circ}{\mathbb P}}\Big\{ \Lambda^{\gamma_1^o, u_2}(x_0, y_1)  \Big[(-2)(x_0+\gamma_1^o(x_0) \nonumber \\  -u_2)\Big)  \Big] 
\Big\vert y_1 \Big\}\Big\vert_{u_2=\gamma_2^o(y_1)} =0, \hso \sr{\circ}{\mathbb P}\big\vert_{ y_1}-a.s.
\label{ws15}
\end{flalign}
Since for any RV $X \in L^1(\Omega, {\cal F}, {\mathbb P}^{\gamma^o}) $,  conditional expectations with respect to a sub-sigma algebra ${\cal G}\subset {\cal F} $,  are related by the  reverse measure transformation,     ${\bf  E}^{{\mathbb P}^{\gamma^o}}\big\{   X\big\vert {\cal G} \big\}=\frac{  {\bf  E}^{\sr{\circ}{\mathbb P}}\big\{ \frac{ {\mathbb P}^{\gamma^o}(d\omega)}{\sr{\circ}{\mathbb P}(d\omega)} X
\big\vert {\cal G} \big\}}{ {\bf E}^{\sr{\circ}{\mathbb P}}\big\{ \frac{ {\mathbb P}^{\gamma^o}(d\omega)}{\sr{\circ}{\mathbb P}(d\omega)} 
\big\vert {\cal G} \big\}}-$a.s., with  ${\mathbb E}^{\sr{\circ}{\mathbb P}}\big\{ \frac{ {\mathbb P}^{\gamma^o}(d\omega)}{\sr{\circ}{\mathbb P}(d\omega)} 
\big\vert {\cal G} \big\}>0-$a.s., from 
 \eqref{ws14}, \eqref{ws15} we obtain  \eqref{ws16}, \eqref{ws17}.  Then from \eqref{ws16}, \eqref{ws17} by simple algebra we obtain \eqref{nl_1_th},  \eqref{nl_2_th}.  
%
%
%
%
\end{proof}


%

\subsection{Fixed Point of Optimal Strategies of the Counterexample }
\label{sec:existenceanduniqueness}
For Problem \# 1,    $\pi^{{\mathbb P}^{\gamma}}(k^2,{\bf P}_{x_0}, G(0,\sigma^2))$,  we  show that   the optimal strategies $(\overline{\gamma}_1^o, \gamma_2^o)$ are fixed point solutions of the    integral equations (\ref{eq:gamma1bar_in}),   (\ref{eq:gamma2_in}), thus establishing existence and uniqueness of solutions in appropriate spaces.  

 
Consider Theorem~\ref{thm;w-sp-1}, and define the nonlinear integral operator  by,
\begin{align}
&F:L^2\times L^2\rightarrow L^2\times L^2, \hso F(\overline{\gamma}_1,\gamma_2)\tri  \begin{pmatrix}
\overline{\gamma}_1(x_0) \\
\gamma_2(y_1)
\end{pmatrix},   \label{fix-1}   \\
&\mbox{$(\overline{\gamma}_1(x_0),\gamma_2(y_1))$ satisfy equations (\ref{eq:gamma1bar_in}),   (\ref{eq:gamma2_in}).}
\label{eqq1}
\end{align}
Define,  
\begin{align*}
&f(x_0, \overline{\gamma}_1,\gamma_2)\tri  \begin{pmatrix}
f_1(x_0, \overline{\gamma}_1,\gamma_2)\\
f_2(x_0, \overline{\gamma}_1,\gamma_2)
\end{pmatrix}, \hso \mbox{where} \\
&   f_1(x_0, \overline{\gamma}_1,\gamma_2) =  -\frac{1}{k^2} 
   \Big\{\frac{1}{2\sigma^2}\big(\zeta-\overline{\gamma}_1(x_0)\big)\big(\overline{\gamma}_1(x_0)-\gamma_2(\zeta)\big)^2\\
&+\big(\overline{\gamma}_1(x_0)-\gamma_2(\zeta)\big)\Big\}\frac{1}{\sqrt{2\pi\sigma}}\exp\Big\{-\frac{\big(\zeta-\overline{\gamma}_1(x_0)\big)^2}{2\sigma^2}\Big\} \\
  &  f_2(x_0, \overline{\gamma}_1,\gamma_2) = \overline{\gamma}_1(\xi) \frac{   \exp\Big\{- \frac{ \big(y_1 -\overline{\gamma}_1(\xi)\big)^2}{2 \sigma^2}\Big\}     }{    \int_{-\infty}^\infty  \exp\Big\{- \frac{ \big(y_1 -\overline{\gamma}_1(\xi)\big)^2}{2 \sigma^2}\Big\} {\bf P}_{x_0}(d\xi)    }.
\end{align*}
The nonlinear operator $F(\overline{\gamma}_1,\gamma_2)$ can be expressed as,
\begin{align}
&F(\overline{\gamma}_1,\gamma_2)=
\begin{pmatrix}
x_0\\
0
\end{pmatrix} -\int f(x_0,\overline{\gamma}_1,\gamma_2) \circ  d\mu, \label{integ}\\
&d\mu=\begin{pmatrix}
d\zeta \\
{\bf P}_{x_0}(d\xi)
\end{pmatrix} \label{fix-2}
\end{align}
where the notation "$\circ$" indicates  that $f_1$ is integrated w.r.t. $d \zeta$ and $f_2$ w.r.t. to ${\bf P}_{x_0}(d\xi)$.  
Note that $f(x_0,\cdot)$ is continuously differentiable in $(\overline{\gamma}_1,\gamma_2)$.

\begin{theorem} (Fixed point solution of optimal strategies of the counterexample,   Problem \# 1,    $\pi^{{\mathbb P}^{\gamma}}(k^2,{\bf P}_{x_0}, G(0,\sigma^2))$)\\
Consider Problem \# 1,    $\pi^{{\mathbb P}^{\gamma}}(k^2,{\bf P}_{x_0}, G(0,\sigma^2))$, and the optimal strategies  $(\overline{\gamma}_1,\gamma_2)$ satisfying the  two integral equations of Theorem~\ref{thm;w-sp-1}.(iii), i.e., \eqref{nl_1_th},  \eqref{nl_2_th}. \\
The following hold. \\
(i) The nonlinear integral operator defined by (\ref{fix-1})-(\ref{fix-2})  is Frechet differentiabe in $\overline{\gamma}_1$ and $\gamma_2$,  i.e.,  there exists a continuous linear operator $L_{\overline{\gamma}_1,\gamma_2}:L^2\times L^2\rightarrow L^2\times L^2 $ such that for all $h \in L^2\times L^2,\, h=\begin{pmatrix}
h_1\\
h_2
\end{pmatrix} $,  
\begin{align}
&\frac{\| F(\overline{\gamma}_1+h_1,\gamma_2+h_2)-F(\overline{\gamma}_1,\gamma_2)-L_{\overline{\gamma}_1,\gamma_2}h\|_2}{\|h\|_2} \rightarrow 0 \hspace{3mm} \nonumber \\ 
&{\rm as} \hspace{4mm} {\|h\|_2}\rightarrow 0. \label{fix-3}
\end{align}
(ii)  There exists a unique fixed point solution  $(\overline{\gamma}_1^o,\gamma_2^o)\in L^2 \times L^2$ of the integral  equations (\ref{eq:gamma1bar_in}),   (\ref{eq:gamma2_in}). 
\end{theorem}
\begin{proof} (i)
First, we note that 
for $h_1\in L^2, h_2 \in L^2$,
\begin{align}
F(\overline{\gamma}_1+h_1,\gamma_2+h_2)-F(\overline{\gamma}_1,\gamma_2)= \nonumber \\ 
\int \Big[f(x_0,\overline{\gamma}_1+h_1,\gamma_2+h_2) -f(x_0,\overline{\gamma}_1,\gamma_2)\Big]  \circ d\mu . \label{eq:3}
\end{align}
To show that $F(\overline{\gamma}_1,\gamma_2)$ is Frechet differentiable, we need to show that there exists a continuous linear operator $L_{\overline{\gamma}_1,\gamma_2}:L^2\times L^2\rightarrow L^2\times L^2 $ such that for all $h \in L^2\times L^2,\, h=\begin{pmatrix}
h_1\\
h_2
\end{pmatrix} $ we have that (\ref{fix-3}) holds. 
In this case, we put $F'(\overline{\gamma}_1,\gamma_2)=L_{\overline{\gamma}_1,\gamma_2}$. By Lagrange's theorem \cite{Lusternik} as $\|h\|_2 \rightarrow0$ we have 
\begin{align}
f(x_0,\overline{\gamma}_1+h_1,&\gamma_2+h_2)-f(x_0,\overline{\gamma}_1,\gamma_2)\nonumber \\
&\rightarrow \nabla_{(\overline{\gamma}_1, \gamma_2) } f   (x_0,\overline{\gamma}_1,\gamma_2)h \label{eq:5}
\end{align}
where $ \nabla_{(\overline{\gamma}_1, \gamma_2)} f (x_0,\cdot)$
 is the partial derivative of $f(x_0, \cdot)$ with respect to $ (\overline{\gamma}_1, \gamma_2)$, 
and the derivatives are easily computed.
given as follows. 
\begin{align}
&\nabla_{\overline{\gamma}_1}     f_1=-\frac{1}{k^2}\Big\{-\frac{1}{2\sigma^2}(\overline{\gamma}_1-\gamma_2)^2+\frac{1}{\sigma^2}(\xi-\overline{\gamma}_1)(\overline{\gamma}_1-\gamma_2) +1 \Big\} \nonumber \\
&.\frac{1}{\sqrt{2\pi}\sigma}\exp\Big\{-\frac{(\xi-\overline{\gamma}_1)^2}{2\sigma^2}\Big\} -\frac{1}{k^2}\Big\{\frac{1}{2\sigma^2}(\xi-\overline{\gamma}_1)(\overline{\gamma}_1-\gamma_2)^2 \nonumber \\
&+(\overline{\gamma}_1-\gamma_2)\Big\} \left( \frac{\xi-\overline{\gamma}_1}{2\sigma^2}\right) \frac{1}{\sqrt{2\pi}\sigma}\exp\Big\{-\frac{(\xi-\overline{\gamma}_1)^2}{2\sigma^2}\Big\}.
\end{align}
\begin{align}
& \nabla_{\gamma_2}  f_1=-\frac{1}{k^2}\Big\{-\frac{1}{\sigma^2}(\xi-\overline{\gamma}_1)(\overline{\gamma}_1-\gamma_2)-1\Big\} \nonumber \\ 
&.\frac{1}{\sqrt{2\pi}\sigma}\exp\Big\{-\frac{(\xi-\overline{\gamma}_1)^2}{2\sigma^2}\Big\}, \\
\nabla_{\overline{\gamma}_1}  f_2&=\frac{\exp \Big \{-\frac{(y_1-\overline{\gamma}_1)^2}{2\sigma^2} \Big\}}{\int_{\infty}^{\infty}\exp \Big \{-\frac{(y_1-\overline{\gamma}_1)^2}{2\sigma^2}\Big\}{\bf P}_{x_0}(d\xi)} \nonumber \\
&+\overline{\gamma}_1\Bigg[ \frac{1}{\sigma^2}(y_1-\overline{\gamma}_1)\exp \Big \{-\frac{(y_1-\overline{\gamma}_1)^2}{2\sigma^2}\Big\} \nonumber\\
&.\int_{-\infty}^{\infty}\exp \Big \{-\frac{(y_1-\overline{\gamma}_1)^2}{2\sigma^2}\Big\}{\bf P}_{x_0}(d\xi) \nonumber \\
&-\exp \Big \{-\frac{(y_1-\overline{\gamma}_1)^2}{2\sigma^2}\Big\} \int_{-\infty}^{\infty}\frac{1}{\sigma^2}(y_1-\overline{\gamma}_1) \nonumber\\
&. \exp \Big \{-\frac{(y_1-\overline{\gamma}_1)^2}{2\sigma^2}\Big\}{\bf P}_{x_0}(d\xi) \Bigg]\Big/ \nonumber \\ 
&\left( \int_{-\infty}^{\infty}\exp \Big \{-\frac{(y_1-\overline{\gamma}_1)^2}{2\sigma^2}\Big\}{\bf P}_{x_0}(d\xi) \right)^2,\\
\nabla_{\gamma_2} f_2 &= 0. 
\end{align}
Expressions (\ref{eq:3}) and (\ref{eq:5}) imply that:
\begin{align}
F'(\overline{\gamma}_1,{\gamma}_2) h =L_{{\overline{\gamma}_1,\gamma_2}}h=\int \nabla_{(\overline{\gamma}_1, \gamma_2)} f(x_0,\overline{\gamma}_1,\gamma_2) h\circ  d\mu .
\end{align}
Thus $L_{{\overline{\gamma}_1,\gamma_2}}$ is an integral operator from $L^2\times L^2$ into $L^2\times L^2$ with kernel $\nabla_{(\overline{\gamma}_1, \gamma_2)}f(x_0,\overline{\gamma}_1,\gamma_2)$. $L_{\overline{\gamma}_1,\gamma_2}$ is clearly a continuous bounded linear operator, since the induced norm satisfies
\begin{align}
& \|L_{\overline{\gamma}_1,\gamma_2}\| = \sup_{\substack{\|h\|_2\leq1\\h \in L^2\times L^2}}\Big \|\int_{- \infty}^\infty \nabla_{(\overline{\gamma}_1, \gamma_2)} f(x_0,\overline{\gamma}_1,\gamma_2)h \circ d\mu\Big\|_2 \nonumber \\
&\leq  \left\|  \begin{pmatrix} \int_{- \infty}^\infty \left[ \left(\nabla_{\overline{\gamma}_1}     f_1\right)^2 + \left(\nabla_{\gamma_2}     f_1\right)^2 \right] d\zeta \\ 
\int_{- \infty}^\infty \left[ \left(\nabla_{\overline{\gamma}_1}     f_2\right)^2 + \left(\nabla_{\gamma_2}     f_2\right)^2 \right] {\bf P}_{x_0}(d\xi) \end{pmatrix} \right\|_2 < \infty
\end{align}
Moreover, since $F(\overline{\gamma}_1,\gamma_2)$ is a continuous differentiable (in the sense of Frechet) operator from $L^2\times L^2$ into $L^2\times L^2$, then  it satisfies the Lipschitz condition:
\begin{align*}
\|F(\tilde{\overline{\gamma}}_1,\tilde{\gamma}_2)-F(\overline{\gamma}_1,\gamma_2)\|_2 \leq \ell\left(\|\tilde{\overline{\gamma}}_1-\overline{\gamma}_1\|_2+\|\tilde{\gamma}_2-\gamma_2\|_2 \right) \\
\ell \tri \sup_{0\leq\theta\leq1}\Big\|F'\Big( \theta\tilde{\overline{\gamma}}_1+(1-\theta)\gamma_1, \theta\tilde{\gamma}_2+(1-\theta)\gamma_2\Big) \Big\|
\end{align*}
Note that the Lipschitz constant can be made less than 1 by either increasing $k^2$ in the payoff function or by using  a weighted $L^2\times L^2$-norm by simply dividing by $(\ell+1)$, for example. The expressions of the two integral equations  follow from the contraction principle,  which guarantees the existence of a fixed point  $( \overline{\gamma}_1,\gamma_2)$ for the nonlinear operator $F(\overline{\gamma}_1,\gamma_2)$.

Alternatively, existence and uniqueness of solutions of the two integral equations (\ref{eq:gamma1bar_in}),   (\ref{eq:gamma2_in}) can also be proven by invoking the inverse function theorem in Banach spaces \cite{Lusternik}, as follows.
By  using  (\ref{integ}) we write,
\begin{eqnarray}
\begin{pmatrix}
\overline{\gamma}_1^o(x_0)\\   
\gamma_2^o(y_1)  
\end{pmatrix} +
\int f(x_0,\overline{\gamma}_1,\gamma_2) \circ d\mu =
\begin{pmatrix}
x_0\\
0
\end{pmatrix}.  
\label{IVTexp}
\end{eqnarray}
Moreover,  $f(x_0, 0,0)  = 0, \nabla_{(\overline{\gamma}_1, \gamma_2)} f(x_0,0,0) \neq 0$. 
By straightforward computations, we can verify that unity is not an eigenvalue of the kernel $\nabla_{(\overline{\gamma}_1, \gamma_2)}f(x_0,0,0)$, i.e., the linear integral equation $z   + \int \nabla_{(\overline{\gamma}_1, \gamma_2)}f(x_0,0,0) z d\mu = 0 $
does not have a non vanishing solution.  Now note that the left-hand-side (LHS) of (\ref{IVTexp}) vanishes for 
$\overline{\gamma}_1=\gamma_2 \equiv 0$, and the derivative $F'(\overline{\gamma}_1,\gamma_2)$ exists in a neighborhood of $0$ is bounded and continuous there.  By virtue of the inverse function theorem \cite{Lusternik} the inverse $[F'(0,0)]^{-1}$ exists, and therefore (\ref{IVTexp}) admits a unique solution for every $x_0$ sufficiently small.
\end{proof}


\section{Numerical Evaluation of PbP Strategies of the Counterexample}
\label{sect:num-eva}
In this Section, we determine the  optimal strategies $(\overline{\gamma}_1^o, \gamma_2^o)$ and corresponding payoff by  invoking a  numerical integration method to solve the    integral equations (\ref{eq:gamma1bar_in}),   (\ref{eq:gamma2_in}), and     
  compare our findings  to other  payoffs  reported in the literature. 

\subsection{Numerical Integration of Nonlinear Integral Equations}
\label{subsec::numericalimplementation}
Since  the  exponential function in  (\ref{eq:gamma1bar_in})   is  Gaussian, we employ the Gauss Hermite Quadrature (GHQ) method.  
First, we briefly review the GHQ method. The approximate numerical integration formula for a function $f(x)$ with values in  $(-\infty,\infty)$ with the weight function $e^{-x^2}$ is \cite{GHQGreenwood}:

\begin{equation}
    \int_{-\infty}^{\infty} f(x) e^{-x^2} dx \approx \Sigma_{i=1}^n f(x_{i,n}) \lambda_{i,n}
    \label{eq:GaussQuadratureRule}
\end{equation}
\noindent where the abscissas $\{x_{i,n}\}$ are the roots of the $n^{th}$ order Hermite polynomial
\begin{align*}
    H_n(x) = -\sqrt{2}^n h_n( \sqrt{2}x) = 0, \hso h_n(x) = e^{\frac{x^2}{2}} \frac{d^n(e^{\frac{-x^2}{2}})}{dx^n}
\end{align*}
where 
 the weights $\{\lambda_{i,n}\}$ are given by
\begin{align}
    \lambda_{i,n} = \frac{\sqrt{\pi}2^{n+1}{n!}}{H_n^{'}(x_{i,n})^2}, \hso H_n^{'}(x) = 2nH_{n-1}(x).
\end{align}
For $n\leq 10$ and higher orders, the zeros $x_{i,n}$ of the Hermite polynomial $H_n(x)$ and the weights $\lambda_{i,n}$ are calculated in \cite{GHQGreenwood,MaxwellGHQ}.   By  \cite{GolubGHQ}  the Gauss quadrature rule \eqref{eq:GaussQuadratureRule} is exact for all continuous 
$f$ that are polynomials of degree $\leq 2n-1$. 


Consider the first strategy  (\ref{eq:gamma1bar_in}) and the change of variables as $z = {\frac{\zeta-\overline{\gamma}_1(x_0)}{\sqrt{2\sigma^2}}}$ and $dz = {\frac{d\zeta}{\sqrt{2\sigma^2}}}$. Then,
\begin{equation*}
    \begin{split}
      &  \overline{\gamma}_1(x_0) = x_0 - \frac{1}{\sqrt{\pi}k^2} \bigintssss_{-\infty}^{\infty} \Big\{\frac{z}{\sqrt{2\sigma^2}}  (\overline{\gamma}_1(x_0) - 
        \gamma_2(\sqrt{2\sigma^2}z \\
    &    + \overline{\gamma}_1(x_0)))^2 +  (\overline{\gamma}_1(x_0) - \gamma_2(\sqrt{2\sigma^2}z + \overline{\gamma}_1(x_0)))\Big\} e^{-z^2} dz
    \end{split}
\end{equation*}
Using GHQ approximation \eqref{eq:GaussQuadratureRule},
\begin{equation}
    \begin{split}
      &  \overline{\gamma}_1(x_0) \approx  x_0 - \frac{1}{\sqrt{\pi}k^2} \sum_{i=1}^n \Big\{\frac{z_i}{\sqrt{2\sigma^2}}   (\overline{\gamma}_1(x_0) - \gamma_2(\sqrt{2\sigma^2}z_i   \\
        &+ \overline{\gamma}_1(x_0)))^2    + (\overline{\gamma}_1(x_0) - \gamma_2(\sqrt{2\sigma^2}z_i + \overline{\gamma}_1(x_0)))\Big\} \lambda_i. 
    \end{split}
    \label{eq:gamma1bar_GHQ}
\end{equation}
Similarly, for  (\ref{eq:gamma2_in})  with the change of variable $z = \frac{\xi}{\sqrt{2\sigma_x^2}}$,  then
\begin{align}
    &\gamma_2(y_1) = \frac{\bigintssss_{-\infty}^{\infty} \overline{\gamma}_1(\xi) \exp{(-\frac{(y_1-\overline{\gamma}_1(\xi))^2}{2\sigma^2})}  \exp{(-\frac{\xi^2}{2\sigma_x^2})} d\xi}{\int_{-\infty}^{\infty} \exp{(-\frac{(y_1-\overline{\gamma}_1(\xi))^2}{2\sigma^2})} \exp{(-\frac{\xi^2}{2\sigma_x^2})} d\xi} \nonumber \\
    &= \frac{\bigintssss_{-\infty}^{\infty} \overline{\gamma}_1(\sqrt{2\sigma_x^2} z) \exp{(-\frac{(y_1-\overline{\gamma}_1(\sqrt{2\sigma_x^2} z))^2}{2\sigma^2})}  e^{-z^2} \sqrt{2\sigma_x^2} dz}{\int_{-\infty}^{\infty} \exp{(-\frac{(y_1-\overline{\gamma}_1(\sqrt{2\sigma_x^2} z))^2}{2\sigma^2})} e^{-z^2} \sqrt{2\sigma_x^2} dz} \nonumber \\
    & \approx  \frac{\sum_{i=1}^{n} \overline{\gamma}_1(\sqrt{2\sigma_x^2} z_i) \exp{(-\frac{(y_1-\overline{\gamma}_1(\sqrt{2\sigma_x^2} z_i))^2}{2\sigma^2})} \lambda_i}{\sum_{i=1}^{n} \exp{(-\frac{(y_1-\overline{\gamma}_1(\sqrt{2\sigma_x^2} z_i))^2}{2\sigma^2})} \lambda_i}.
    \label{eq:gamma2GHQ}
\end{align}
Consider \eqref{eq:gamma1bar_GHQ},  since $z_i$ and $\lambda_i$ are the (known) nodes and weights, for certain $x_0 \in \mathbb{R}$, the unknowns are $\overline{\gamma}_1(x_0)$ and $\gamma_2(\sqrt{2\sigma^2}z_i + \overline{\gamma}_1(x_0)))$ (whose argument is in turn a function of $\overline{\gamma}_1(x_0)$). In order to solve this equation, we can employ the expression for $\gamma_2(y_1)$ from \eqref{eq:gamma2GHQ} with  $y_1 = \sqrt{2\sigma^2}z_i + \overline{\gamma}_1(x_0)$,
\begin{align}
    & \gamma_2(\sqrt{2\sigma^2}z_i + \overline{\gamma}_1(x_0))) 
     \approx \bigg( \sum_{i=1}^{n} \big( \overline{\gamma}_1(\sqrt{2\sigma_x^2} z_i) \label{eq:gamma2ofalot}  \\
     & \exp{(-\frac{(\sqrt{2\sigma^2}z_i + \overline{\gamma}_1(x_0)-\overline{\gamma}_1(\sqrt{2\sigma_x^2} z_i))^2}{2\sigma^2})} \lambda_i \big) \bigg) \bigg/ \nonumber \\
     & \bigg( \sum_{i=1}^{n} \Big( \exp{(-\frac{(\sqrt{2\sigma^2}z_i + \overline{\gamma}_1(x_0)-\overline{\gamma}_1(\sqrt{2\sigma_x^2} z_i))^2}{2\sigma^2})} \lambda_i \Big) \bigg). \nonumber
\end{align}
Substituting  $\gamma_2(\sqrt{2\sigma^2}z_i + \overline{\gamma}_1(x_0)))$ from \eqref{eq:gamma2ofalot} in (\ref{eq:gamma1bar_GHQ}), then 
\begin{align}
    \begin{split}
       & \overline{\gamma}_1(x_0) \approx x_0 - \frac{1}{\sqrt{\pi}k^2} \sum_{i=1}^n \lambda_i \Bigg\{\frac{z_i}{\sqrt{2\sigma^2}} \\ &\Bigg(\overline{\gamma}_1(x_0) - 
        \bigg( \sum_{j=1}^{n} \big( \overline{\gamma}_1(\sqrt{2\sigma_x^2} z_j)  \\
      &\exp{(-\frac{(\sqrt{2\sigma^2}z_i + \overline{\gamma}_1(x_0)-\overline{\gamma}_1(\sqrt{2\sigma_x^2} z_j))^2}{2\sigma^2})} \lambda_j \big) \bigg) \bigg/  \\
      &\bigg( \sum_{j=1}^{n} \Big( \exp{(-\frac{(\sqrt{2\sigma^2}z_i + \overline{\gamma}_1(x_0)-\bar{\gamma}_1(\sqrt{2\sigma_x^2} z_j))^2}{2\sigma^2})} \lambda_j \Big) \bigg)\Bigg)^2 \\
       & + \Bigg(\overline{\gamma}_1(x_0) - \bigg( \sum_{j=1}^{n} \big( \overline{\gamma}_1(\sqrt{2\sigma_x^2} z_j) \\
      &\exp{(-\frac{(\sqrt{2\sigma^2}z_i + \overline{\gamma}_1(x_0)-\overline{\gamma}_1(\sqrt{2\sigma_x^2} z_j))^2}{2\sigma^2})} \lambda_j \big) \bigg) \bigg/ \\
      &\bigg( \sum_{j=1}^{n} \Big( \exp{(-\frac{(\sqrt{2\sigma^2}z_i + \overline{\gamma}_1(x_0)-\overline{\gamma}_1(\sqrt{2\sigma_x^2} z_j))^2}{2\sigma^2})} \lambda_j \Big) \bigg)\Bigg)\Bigg\} .
    \end{split}
    \label{eq:bignonlineareq}
\end{align}
While $x_0 \in \mathbb{R}$ and $\sqrt{2\sigma_x^2} z_i$ are known, $\overline{\gamma}_1(x_0)$ and $\overline{\gamma}_1(\sqrt{2\sigma_x^2} z_i)$ are unknown. Let $s_i = \overline{\gamma}_1(\sqrt{2\sigma_x^2} z_i), \ \forall i$. Then  \eqref{eq:bignonlineareq}  contains $(n+1)$  unknowns, i.e., $n$ $s_i$'s and $\overline{\gamma}_1(x_0)$. 
\begin{align*}
    & \overline{\gamma}_1(x_0) \approx x_0 - \frac{1}{\sqrt{\pi}k^2} \sum_{i=1}^n \lambda_i \Bigg\{\frac{z_i}{\sqrt{2\sigma^2}} \\ &\Bigg(\overline{\gamma}_1(x_0) - 
        \bigg( \sum_{j=1}^{n} \big( s_j 
      \exp{(-\frac{(\sqrt{2\sigma^2}z_i + \overline{\gamma}_1(x_0)-s_j)^2}{2\sigma^2})} \lambda_j \big) \bigg) \bigg/  \\
      &\bigg( \sum_{j=1}^{n} \Big( \exp{(-\frac{(\sqrt{2\sigma^2}z_i + \overline{\gamma}_1(x_0)-s_j)^2}{2\sigma^2})} \lambda_j \Big) \bigg)\Bigg)^2 +\\
       & \Bigg(\overline{\gamma}_1(x_0) - \bigg( \sum_{j=1}^{n} \big( s_j
      \exp{(-\frac{(\sqrt{2\sigma^2}z_i + \overline{\gamma}_1(x_0)-s_j)^2}{2\sigma^2})} \lambda_j \big) \bigg) \bigg/ \\
      &\bigg( \sum_{j=1}^{n} \Big( \exp{(-\frac{(\sqrt{2\sigma^2}z_i + \overline{\gamma}_1(x_0)-s_j)^2}{2\sigma^2})} \lambda_j \Big) \bigg)\Bigg)\Bigg\} 
\end{align*}
Substituting $x_0 = x_{0l} = \sqrt{2\sigma_x^2} z_l$ for each $l \in \{1,2,\hdots,n\}$, we obtain $n$ nonlinear equations with $n$ $s_l$'s that are unknown, given in \eqref{eq:sysnonlineareq}. Each $s_l$, which is the value of $\overline{\gamma}_1(x_0)$ at nodes selected according to GHQ, is the signaling level of the control action. Rearranging \eqref{eq:sysnonlineareq} to move all terms on one side, we denote the resulting system of nonlinear equations as $f_{sysnonlin}:\mathbb{R}^n \to \mathbb{R}^n$.
\begin{align}
    & \forall l = {1,2,\hdots,n} \nonumber \\
    & t_l \approx \sqrt{2\sigma_x^2} z_l - \frac{1}{\sqrt{\pi}k^2} \sum_{i=1}^n \lambda_i \Bigg\{\frac{z_i}{\sqrt{2\sigma^2}} \nonumber \\ &\Bigg(t_l - 
        \bigg( \sum_{j=1}^{n} \big( t_j 
      \exp{(-\frac{(\sqrt{2\sigma^2}z_i + t_l-t_j)^2}{2\sigma^2})} \lambda_j \big) \bigg) \bigg/ \nonumber \\
      &\bigg( \sum_{j=1}^{n} \Big( \exp{(-\frac{(\sqrt{2\sigma^2}z_i + t_l-t_j)^2}{2\sigma^2})} \lambda_j \Big) \bigg)\Bigg)^2 + \nonumber \\
       & \Bigg(t_l - \bigg( \sum_{j=1}^{n} \big( t_j
      \exp{(-\frac{(\sqrt{2\sigma^2}z_i + t_l-t_j)^2}{2\sigma^2})} \lambda_j \big) \bigg) \bigg/ \nonumber \\
      &\bigg( \sum_{j=1}^{n} \Big( \exp{(-\frac{(\sqrt{2\sigma^2}z_i + t_l-t_j)^2}{2\sigma^2})} \lambda_j \Big) \bigg)\Bigg)\Bigg\}.
      \label{eq:sysnonlineareq}
\end{align}

The solution of the system of $n$ nonlinear equations \eqref{eq:sysnonlineareq} results in $n$ explicit points, i.e., $n$ signaling levels $s_l^*, \ \forall l=1,2,\hdots ,n,$ such that $||f_{sysnonlin}(s_1^*,s_2^*, \hdots , s_n^*)|| $ is close to zero. Using these $n$ signaling levels, we obtain the value of $\overline{\gamma}_1(x_0) \ \forall x_0,$ by substituting $(s_1^*,s_2^*, \hdots , s_n^*)$ in \eqref{eq:bignonlineareq} which results in one unknown $\overline{\gamma}_1(x_0)$ and solving the resulting nonlinear equation for each $x_0$. This is similar to the collocation method used to solve integral equations, \cite{PiecewiseCollocationAtkinson}. Here,  $x_0 = x_{0l} = \sqrt{2\sigma_x^2} z_l$ for each $l \in \{1,2,\hdots,n\}$ are the collocation points and signaling levels are the value of $\overline{\gamma}_1(x_0)$ at the collocation points. 
To obtain the strategy of the second controller, we substitute the signaling levels $(s_1^*,s_2^*, \hdots , s_n^*)$ in \eqref{eq:gamma2GHQ}. This directly gives the expression for $\gamma_2(y_1)$ which is evaluated at $y_1$. 
Note that $y_1 = \overline{\gamma}_1(x_0) + v$, and hence  the values taken by $y_1$ are dictated by the strategy of the first controller $\overline{\gamma}_1(x_0)$. Once both the strategies $\bar{\gamma}_1$ $\gamma_2$ are obtained, we calculate the total cost $J^{{\mathbb P}}$ from \eqref{eq:costeq}.

The algorithm to compute  strategies \eqref{eq:gamma1bar_in},  \eqref{eq:gamma2_in} is  given below.

\begin{framed}
\noindent Input parameters: $k, \sigma, \sigma_x, n$; 
Input signals: $x_0, v.$
\begin{itemize}
    \item[-] Solve $f_{sysnonlin}$ to obtain the signaling levels $(s_1^*,s_2^*, \hdots , s_n^*)$.
    \item[-] For each $x_0$, compute $\bar{\gamma}_1(x_0)$
    \item[-] For all $y_1  = \overline{\gamma}_1(x_0) + v$, compute $\gamma_2(y_1)$
\end{itemize}
\end{framed}


\textbf{\textit{Implementation aspects:}} We employ the software MATLAB to implement the solution strategies \eqref{eq:gamma1bar_in} and \eqref{eq:gamma2_in}. The command \textit{fsolve} is used to solve the system of nonlinear equations $f_{sysnonlin}$ and \textit{lsqnonlin} to solve for $\overline{\gamma}_1(x_0)$. 

\subsection{Results}
\label{subsec:results}
We employed $600,000$ samples for $x_0$ and $v$ generated according to $G(0,\sigma_x^2)$ and $G(0,\sigma^2)$, respectively. The order of the Hermite polynomial in GHQ method is $n=7$. 
The total cost is
\begin{equation}
    J^{{\mathbb P}^\gamma}(\overline{\gamma}_1,\gamma_2) ={\bf E}^{{\mathbb P}^\gamma}\big\{k^2 (\overline{\gamma_1}(x_0) - x_0)^2 + (\overline{\gamma_1}(x_0) - \gamma_2(y_1))^2 \big\}. \label{eq:TotalCost}
\end{equation}
We denote by $J^o=J^{{\mathbb P}^\gamma}(\overline{\gamma}_1,\gamma_2)$ the cost corresponding  to strategies  \eqref{eq:gamma1bar_in} and \eqref{eq:gamma2_in},   implemented using the GHQ method 
in Section \ref{subsec::numericalimplementation}. We consider different values for parameters $k, \sigma_x$, and $\sigma$ and compare $J^o$ with $J^{aff}$, $J^{wit}$ and other previously reported costs. By Lemma 1 of \cite{WitsenhausenOriginal}, the optimal cost satisfies  $J^{{\mathbb P}^\gamma}(\overline{\gamma}_1^o,\gamma_2^o)\leq \min(1,k^2\sigma_x^2)$ (when $\sigma^2=1$). Accordingly, we verify if the cost $J^o$ is less than min$(1,k^2\sigma_x^2)$.

\subsubsection{\textbf{Affine region}}
\label{subsec:parameters111}
As pointed in \cite{WuVerduTransport}, the set of parameter values where $k \nless 0.56$ and $\sigma_x$ is not large, is in the region where affine laws are optimal. We consider the values for the parameters to be $k=1, \sigma_x=1, \sigma = 1$. The optimal control laws \eqref{eq:gamma1bar_in} and \eqref{eq:gamma2_in} are compared with optimal affine laws in Fig. \ref{fig:k1sigmax1}. It is seen that the resulting laws are almost the same as the optimal affine laws. We further compare the cost with $J^{aff}$ and $J^{wit}$ in Table \ref{tab:k1sigmax1}. The negligible difference in $J^{aff}$ and $J^o$ is attributed to numerical inaccuracy in the implementation of \eqref{eq:gamma1bar_in} and \eqref{eq:gamma2_in} through approximate numerical integration method. 

\begin{figure}
    \centering
    \includegraphics[width=\columnwidth]{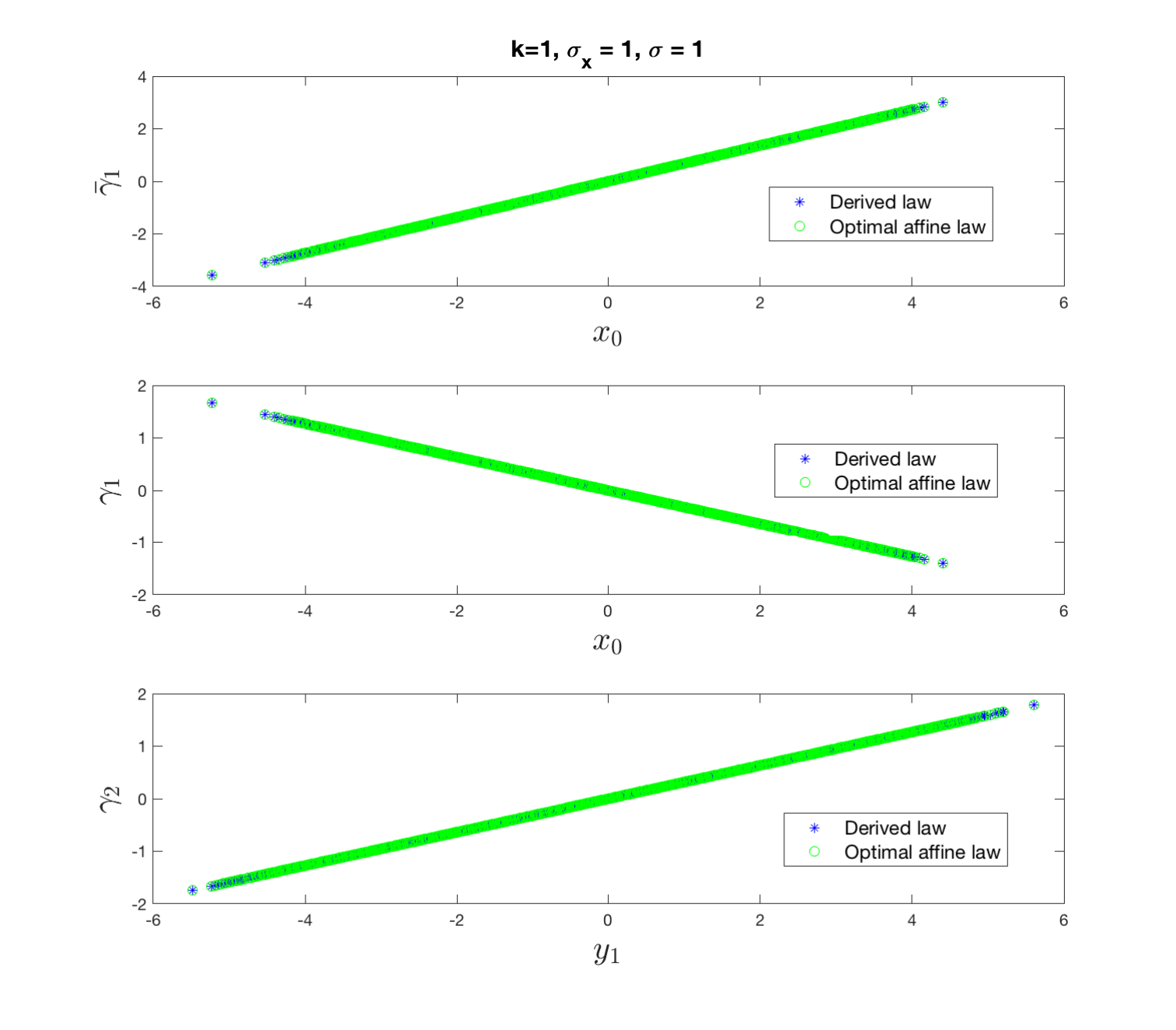}
    \caption{{Comparison of the optimal control laws and 
    the special class of optimal affine laws}}
    \label{fig:k1sigmax1}
\end{figure}

\begin{table}[]
\centering
\begin{tabular}{@{}cccc@{}}
\toprule
          & Stage 1  & Stage 2  & Total Cost \\ \midrule
$J^{aff}$ & $0.1011$ & $0.3174$ & $0.418500414352474$   \\
$J^{wit}$ & $0.4043$ & $0.4480$ & $0.852287449358227$   \\
$J^o$  & $0.1011$ & $0.3174$ & $0.418500469701766$   \\ \bottomrule
\end{tabular}
\caption{Total cost, $k=1, \sigma_x=1$}
\label{tab:k1sigmax1}
\end{table}

\subsubsection{\textbf{Comparison with \cite{McEany1BadPaper}}}
\label{subsec:mceany}
The authors in \cite{McEany1BadPaper} consider three sets of parameter values and find if the corresponding optima are roughly linear or of signaling form. Although the cost obtained is not reported in \cite{McEany1BadPaper}, for the set of parameter values therein, we compare the laws we obtain with the figures therein. $\overline{\gamma}_1(x_0)$ obtained for all the three sets of parameters are shown in Fig. \ref{fig:McEanyparameters}. Consistent with the findings in \cite{McEany1BadPaper}, the first and the third set of parameters result in optima that are linear and nonlinear, respectively. However, the second set of parameters results in linear optima while \cite{McEany1BadPaper} finds the optima to be a mix of linear and signaling form. The corresponding costs are given in Table \ref{tab:McEanyCostComparison}.

\begin{figure}
  \begin{subfigure}[b]{\columnwidth}
    \includegraphics[width=\linewidth]{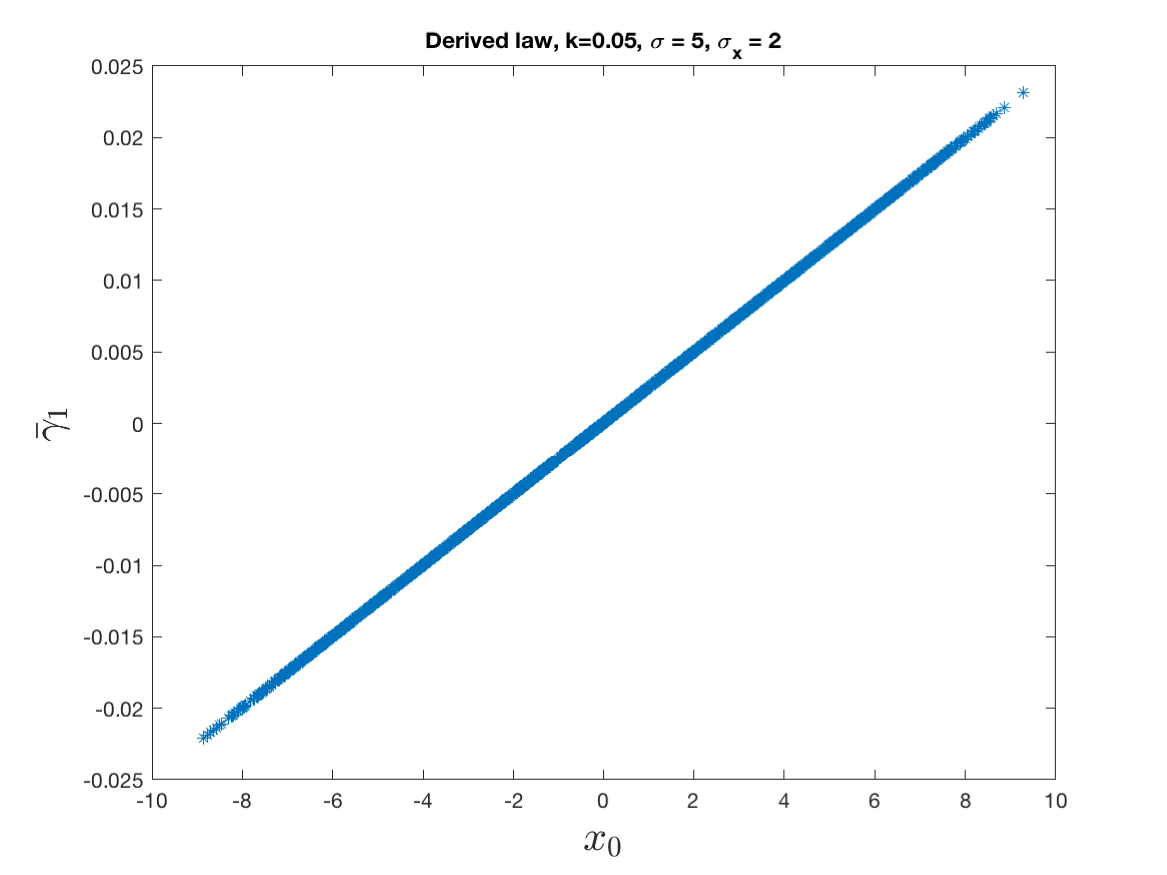}
    \caption{$\overline{\gamma}_1$ is linear for $k=0.05, \sigma = 5, \sigma_x = 2$}
    \label{fig:1}
  \end{subfigure}
  \newline
  \begin{subfigure}[b]{\columnwidth}
    \includegraphics[width=\linewidth]{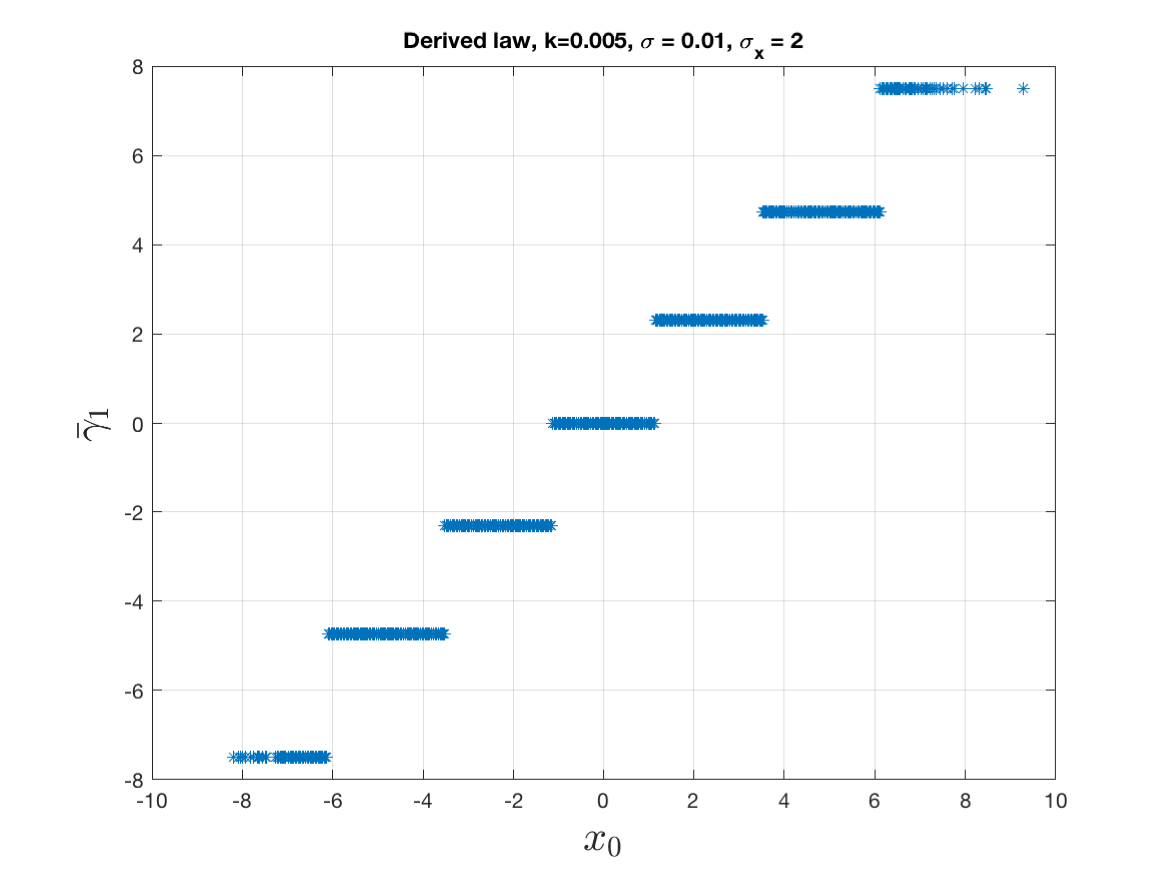}
    \caption{$\overline{\gamma}_1$ is $7-$step non-linear for $k=0.005, \sigma = 0.01, \sigma_x = 2$}
    \label{fig:2}
  \end{subfigure}
  \newline
  \begin{subfigure}[b]{\columnwidth}
    \includegraphics[width=\linewidth]{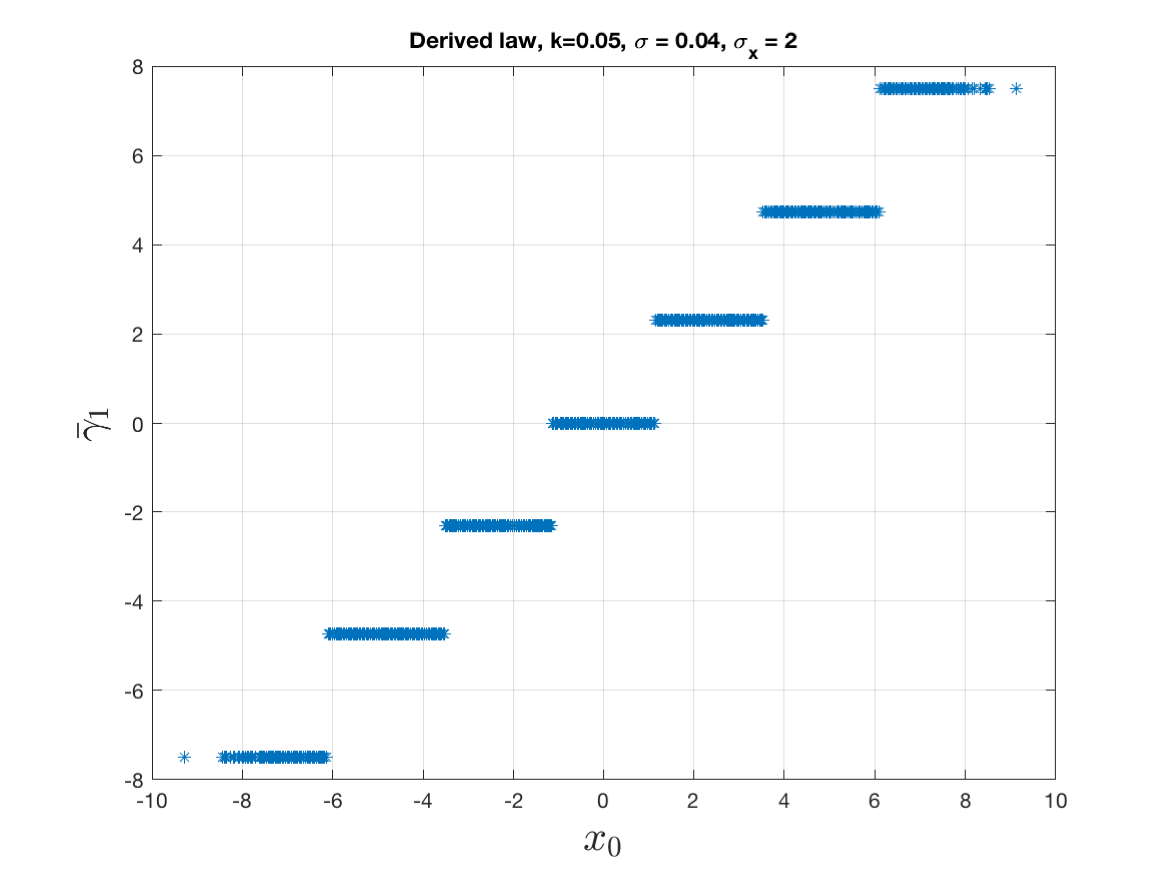}
    \caption{$\overline{\gamma}_1$ is $7-$step non-linear for $k=0.05, \sigma = 0.04, \sigma_x = 2$}
    \label{fig:2}
  \end{subfigure}
  \caption{Optimal first control law $\overline{\gamma}_1$ for parameters in \cite{McEany1BadPaper}}
   \label{fig:McEanyparameters}
\end{figure}

\begin{table}[]
\begin{tabular}{@{}cccc@{}}
\toprule
          & $k=0.05,\sigma=5$ & $k=0.005,\sigma=0.01$   & $k=0.05,\sigma=0.04$ \\
          & $\sigma_x=2$      & $\sigma_x=2$            & $\sigma_x=2$         \\ \midrule
$J^{aff}$ & $0.0100$                & $1.007 \times 10^{-4}$  & $0.0100$                   \\
$J^{wit}$ & $5.2326$                & $4.225 \times 10^{-5}$  & $0.0040$                   \\
$J^o$  & $0.0100$                & $1.1298 \times 10^{-5}$ & $0.0011$                   \\ \bottomrule
\end{tabular}
\caption{Total cost obtained for parameters in \cite{McEany1BadPaper}}
\label{tab:McEanyCostComparison}
\end{table}

\subsubsection{\textbf{Benchmark parameters $k=0.2, \sigma_x=5, \sigma=1$}}
\label{subsec:toughparameters}
The last set of parameters we consider has been the most studied case and has enabled more insights into the solution of the  counterexample. \cite{neuralnetworksolution} provides a numerical solution by employing one-hidden-layer neural network as an approximating network, with corresponding cost  $J^{nn}$. \cite{hierarchialLee} presents a hierarchial search approach where   $\overline{\gamma}_1$ is imposed to be a non-decreasing, step function that is symmetric about the origin. For a considered number of steps, they find the signaling levels (value of $\overline{\gamma}_1$ at the step) and the breakpoints ($x_0$ where the step change occurs). They also find that the cost objective is lower for slightly sloped steps than perfectly leveled steps. Through comparison of their costs for different number of steps, they find that $7-$step solution yields the lowest cost. The cost obtained in \cite{hierarchialLee} is denoted $J^{llh}$ here and the signaling levels therein are $s^{*} = \{0, \pm 6.5, \pm 13.2, \pm 19.9\}$.

In our work, the solution of \eqref{eq:sysnonlineareq} yields the signling levels $s^{**} = \{0, \pm 6.15, \pm 12.8, \pm 19.8\}$ and $||f_{sysnonlin}(s^{**})|| = 10^{-15}$ while $||f_{sysnonlin}(s^{*})|| = 0.7$. Following up on the notes from Section \ref{subsec::numericalimplementation}, Gauss quadrature rule is not exact for the set of parameters $k=0.2, \sigma_x=5, \sigma=1$ because this parameter set lies in the region where the optimal laws are non-linear. Moreover, the optimal non-linear laws are not continuous; they are only piecewise continuous. As a result, the inaccuracy in the approximation using Gauss quadrature rule reveals itself through the system of nonlinear equations $f_{sysnonlin}$. The cost we obtain for signaling levels $s^*$ and $s^{**}$ are $J^o_* = 0.16$ and $J^o_{**} = 0.1712$ respectively. 

The strategy of the first controller, $\overline{\gamma}_1(x_0)$, that we obtain for the signaling levels $s^*$ and $s^{**}$ are shown in Fig \ref{fig:gamma1barsolcompare}. Although we don't externally impose symmetry, it can be observed that $\overline{\gamma}_1$ is symmetric about origin and is non-decreasing. We zoom in on one of the $7$ steps and observe in the left column of Fig \ref{fig:slightlysloped} that the steps are slightly sloped. Further zooming in, we see in the right column of Fig \ref{fig:slightlysloped} that each signaling level is further comprised of a number of closely spaced steps. Similar to this result, authors in \cite{hierarchialLee} added segments in each of the $7$ steps to obtain the cost $J^{llh}=0.167313205338$. We compare both the costs we obtain with previously reported costs in the literature in Table \ref{tab:k0p2sigmax5}. Further in agreement with the findings in \cite{hierarchialLee}, we obtain the lowest cost for $7$ steps, $J_{**}^o = 0.1712$. 

\begin{figure}[h]
    \centering
    \includegraphics[width=\columnwidth]{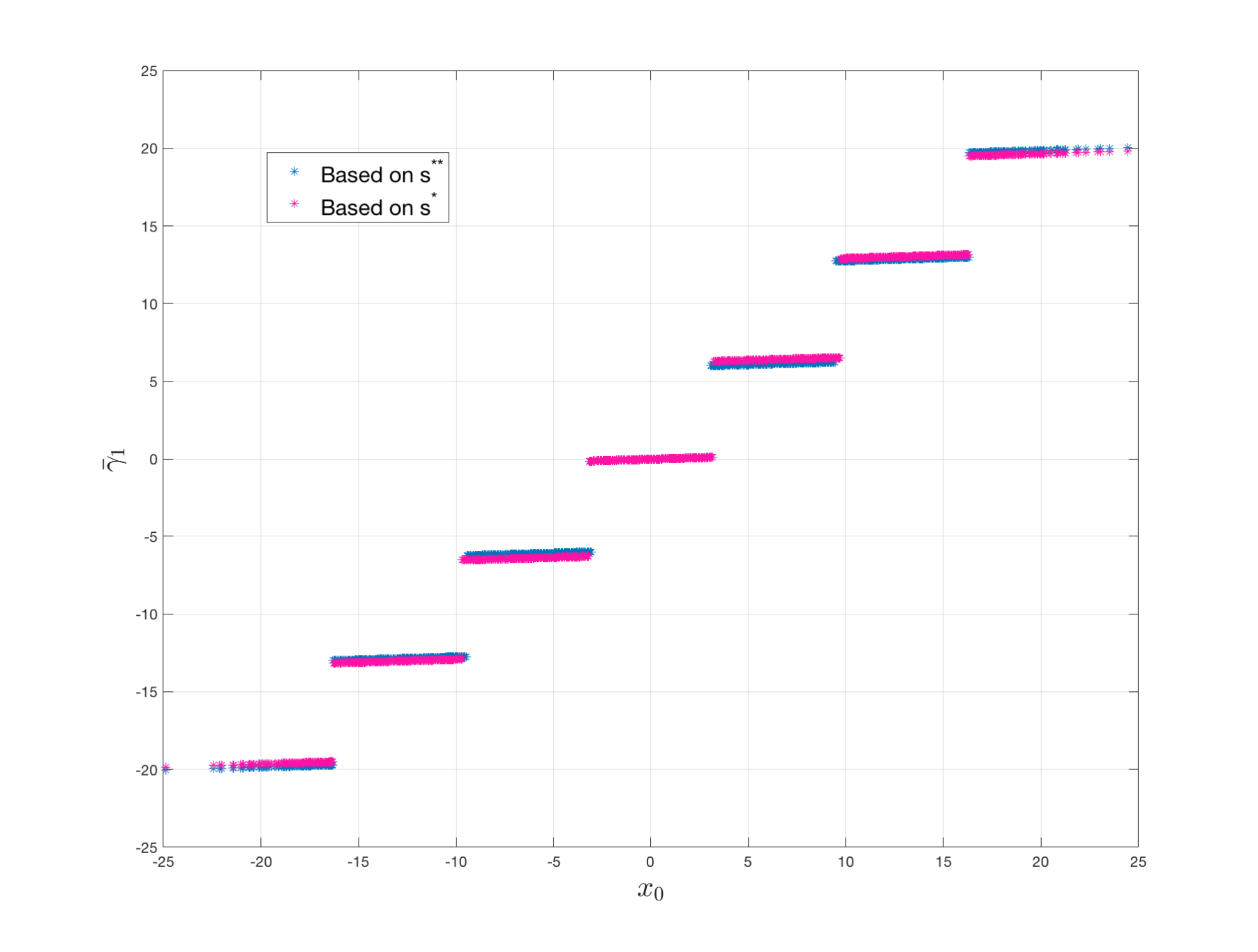}
    \caption{Comparison of the strategy of the first controller for signaling levels $s^*$ and $s^{**}$ for benchmark parameters.}
    \label{fig:gamma1barsolcompare}
\end{figure}

\begin{figure}[h]
    \centering
    \includegraphics[width=\columnwidth]{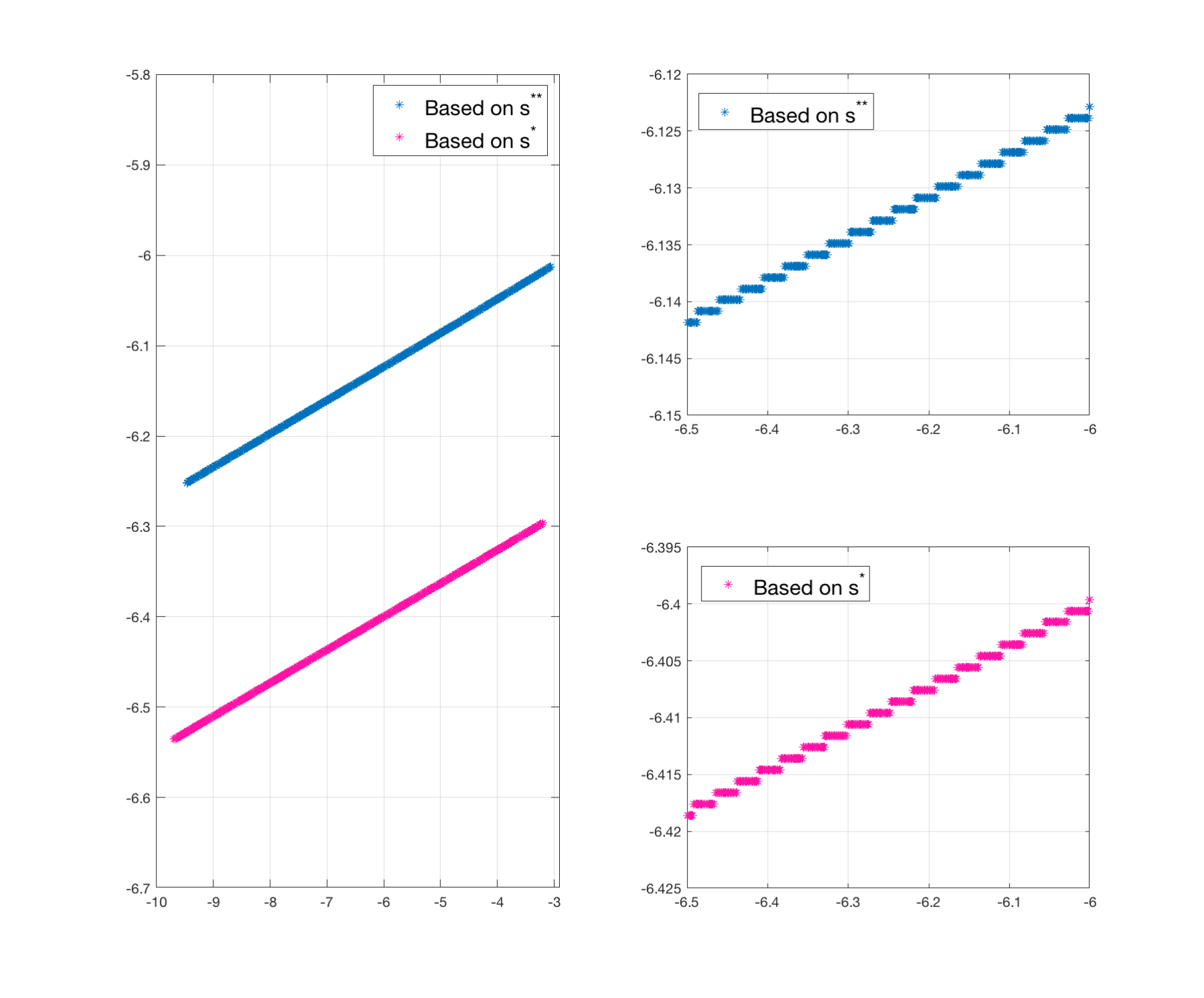}
    \caption{Observation of the slight slope in quantizers -- magnified views of one \quotes{step} Fig. \ref{fig:gamma1barsolcompare}}
    \label{fig:slightlysloped}
\end{figure}

\begin{table*}[t]
\centering
\begin{tabular}{@{}cccc@{}}
\toprule
          & Stage 1             & Stage 2                            & Total Cost          \\ \midrule
$J^{aff}$ & $0.0017428616051158$ & $0.956950417234115$                & $0.958693278839234$ \\
$J^{wit}$ & $0.403507741927546$ & $2.134488364684996 \times 10^{-6}$ & $0.403509876415911$ \\
$J^{nn}$\cite{neuralnetworksolution}  & -                   & -                                  & $0.1735$            \\
$J^{llh}$ \cite{hierarchialLee}  & $0.131884081844$        & $0.035429123524$                       & $0.167313205368$        \\
$J^o_*$  & $0.128541364988695$ & $0.038385613344897$                & $0.166926978333592$ \\
$J^{o}_{**}$  & $0.120110042087359$        & $0.051158481289032$                       & $0.171268523376388$        \\ \bottomrule
\end{tabular}
\caption{Reported and obtained costs, $k=0.2, \sigma_x=5, \sigma=1$}
\label{tab:k0p2sigmax5}
\end{table*}

For the parameters  $k=0.2, \sigma_x=5, \sigma=1$, the number of steps we obtain is same as the value of the Gauss quadrature rule parameter $n$. However, this is not necessarily the case for all parameters;  see Sections \ref{subsec:parameters111} and \ref{subsec:mceany}. The parameter set $k=1, \sigma_x=5, \sigma=1$ is known to lie in a region where the optimal law is affine, and even though we employ $n=7$ steps for GHQ, the resulting control laws are affine. 

\ \

\section{Conclusion}
\label{sec::conclusion}
The paper derives 
 optimality conditions  for general discrete-time  decentralized stochastic dynamic optimal control problems, using  Girsanov's change of measure. The methods is  applied to derive   PbP optimal strategies of Witsenhausen's  counterexample \cite{WitsenhausenOriginal}.    The two strategies are shown to satisfy  nonlinear integral equations, while a fixed point theorem is shown  establishing existence and uniqueness of their solutions. Numerical solutions of the two integral equations are presented and compared to the literature.


An important observation of our investigation of the counterexample  is that, for certain parameter values non-linear strategies out perform  linear strategies,  while for some parameter values linear strategies are indeed PbP optimal\footnote{This observation is consistent with  \cite{WitsenhausenOriginal}, because Theorem~2 in \cite{WitsenhausenOriginal} states that nonlinear strategies outperform affine strategies for certain choices of the problem parameters, and not for all possible choices of parameters.}.  This observation  is not document in previous numerical studies. 



\section{Appendix}
\label{app-1}
In this section, 
we introduce the  basic mathematical concept of change of
probability measure, 
%
known
as Radon-Nikodym derivative Theorem. 

\begin{theorem}\cite{liptser-shiryayev1977,elliott1982}  (Radon-Nikodym Derivative Thm)\\
\label{theorem 1.7.4}
Let $(\Omega, {\cal  F})$ a measurable space and let ${\mathbb  P}$ and  ${\mathbb  Q}$
be two  probability measure defined it. Then ${\mathbb  P}$ is said to be 
absolutely continuous with respect to ${\mathbb  Q}$, denoted by  ${\mathbb  P}\ll {\mathbb  Q}$,
\bes
\mbox{if and only if  $\forall B \in {\cal F}$ such that   ${\mathbb  Q}(B)=0$ then ${\mathbb  P}(B)=0$.} 
\ees
Moreover, ${\mathbb  P}$ is said to be mutually absolutely continuous  with respect to ${\mathbb  Q}$ if and only if  ${\mathbb  P}\ll {\mathbb  Q}$ and  ${\mathbb  Q}\ll {\mathbb  P}$. \\
If ${\mathbb  P}\ll {\mathbb  Q}$ 
then there exists an ${\cal  F}$-measurable function $\phi:\Omega\rar 
{\mathbb R}$,  such that  $\phi \in L^1(\Omega,{\cal  F}, {\mathbb  P})$ and
\bea
{\mathbb  P}(B) \tri \int_Bd {\mathbb  P}(\omega)=\int_B\phi(\omega)d {\mathbb  Q}(\omega), \ \ \forall  B\in {\cal  F}.
\label{eq.211}
\eea
\noi The function $\phi$ is unique except on a subset of
${\mathbb  Q}$-measure zero, and is often written as $\phi \tri \frac{d
{\mathbb  P}}{d {\mathbb  Q}}\Big|_{{\cal  F}}$,   called the Radon-Nikodym
derivative (RND) of ${\mathbb P}$ with respect to (w.r.t.) ${\mathbb Q}$.
\end{theorem}


\begin{theorem}\cite{liptser-shiryayev1977,elliott1982}  (Expectations and Conditional Bayes Rule)\\
\label{thm:RND-c}
Let $(\Omega, {\cal  F})$ be a measurable space,  ${\mathbb  P}$ and
${\mathbb  Q}$ two probability  measures defined on $(\Omega, {\cal  F})$ such that  ${\mathbb  P}\ll {\mathbb  Q}$ and   ${\mathbb  Q}\ll {\mathbb  P}$ 
(mutually absolutely continuous), 
and    $X:\Omega \rar
{\mathbb R}$  a RV such that $X \in L^1(\Omega,{\cal  F}, {\mathbb  P}), X \in L^1(\Omega,{\cal  F}, {\mathbb  Q})$. 
%
Define the RNDs 
\begin{align*}
\phi \tri\frac{d {\mathbb  P}}{d
{\mathbb  Q}}\Big|_{{\cal  F}}\in L^1(\Omega,{\cal  F}, {\mathbb  Q}), \:  \phi^{-1} \tri\frac{d {\mathbb  Q}}{d
{\mathbb  P}}\Big|_{{\cal  F}}\in L^1(\Omega,{\cal  F}, {\mathbb  P}).
\end{align*}
(1) Expectations.  The two probability measures are related by 
\begin{align}
& {\mathbb  P}(B)=\int_B\phi(\omega) d{\mathbb  Q}(\omega),\;
 {\mathbb  Q}(B)=\int_B\phi^{-1}(\omega) d{\mathbb  P}(\omega), \\
&{\bf  E}^{\mathbb  P}\big\{X\big\}={\bf E}^{\mathbb Q}\big\{\phi
X\big\}={\bf E}^{\mathbb Q}\Big\{X \frac{d {\mathbb  P}}{d {\mathbb  Q}}
\Big\}, \\
 &{\bf E}^{\mathbb 
 Q}\big\{X\big\}={\bf  E}^{\mathbb  P}\big\{\phi^{-1}
X\big\}={\bf E}^{\mathbb  P}\Big\{\frac{d {\mathbb  Q}}{d {\mathbb  P}} X\Big\}.
\end{align}
(2) Conditional Bayes Rule.  Let ${\cal G} \subset {\cal  F}$
be  a sub-$\sigma$-field.  
Then,
 \begin{align}
& {\bf E}^{\mathbb   P}\big\{X\big|{\cal G}\big\}(\omega)=\frac{  {\bf E}^{\mathbb Q}\big\{\phi X\big|{\cal G}\big\}}{ {\bf E}^{\mathbb Q}\big\{\phi \big|{\cal G}\big\}}(\omega),
\ \ {\mathbb  P}-a.s., \\
& {\bf E}^{\mathbb Q}\big\{X\big|{\cal G}\big\}(\omega)=\frac{{\bf E}^{\mathbb  P}\big\{\phi^{-1}X\big|{\cal G}\big\}}{{\bf E}^{\mathbb  P}\big\{\phi^{-1}\big|{\cal G}\big\}}(\omega),
\ \ {\mathbb  Q}-a.s.\label{eq.213}
\end{align}
\end{theorem}

%
%
%

\vspace*{-1.5cm}

\begin{IEEEbiography}
{Bhagyashri Telsang}
received her Bachelor of Engineering in Instrumentation and Control from Manipal Institute of Technology in 2014, Master of Science in Systems and Control from Delft University of Technology in 2016, and Doctor of Philosophy in Control Systems from University of Tennessee Knoxville in 2022. Her fields of interest include decentralized control, resource allocation, multi-agent systems and game theory. She is currently a Senior Mechatronics Engineer at ASML in San Diego.
\end{IEEEbiography}

\vspace*{-1.5cm}

\begin{IEEEbiography}
{Seddik M. Djouadi}
received the B.S. degree (Hons.) from the Ecole Nationale Polytechnique, the M.Sc. degree from the University of Montreal, 
and the Ph.D. degree from McGill University, Montreal, all in EE. He is currently a Professor in the EECS Department, 
The University of Tennessee, Knoxville. He has published over about hundred journal and conference papers, some of which were invited papers. His research interests include filtering and control of systems under communication constraints, modeling and control of wireless networks, control systems and applications to autonomous sensor platforms, electromechanical and mobile communication systems, in particular smart grid and power systems, control systems through communication links, networked control systems, and model reduction for aerodynamic feedback flow control. He received a best PES paper award in 2022, the American Control Conference Best Student Paper Certificate (best five in competition), in 1998, the Tibbet Award from AFS Inc., in 1999, 
the Ralph E. Powe Junior Faculty Enhancement Award, in 2005, and the Best Paper Award in the Conference on Intelligent Systems and Automation, in 2008.
\end{IEEEbiography}

\vspace*{-1.0cm}

\begin{IEEEbiography}
{Charalambos D. Charalambous}
received his B.S., M.E., and Ph.D. in 1987,1988, and 1992, respectively, all from the Department of Electrical
Engineering, Old Dominion University, Virginia, USA. In 2003 he joined the
Department of Electrical and Computer Engineering, University of Cyprus. He
was an Associate Professor at University of Ottawa, from 1999 to 2003. He served on the faculty of
McGill University, Department of Electrical and Computer Engineering, as a
non-tenure faculty member, from 1995 to 1999. From 1993 to 1995 he was a
post-doctoral fellow at Idaho State University. He is currently  editor at large of Mathematics of
Control, Signals and Systems and on the editorial board of Entropy  for section, Information Theory, Probability and Statistics.  In the past he served as an Associate Editor
of the IEEE Transactions on Automatic Control, Systems and Control Letters,  and IEEE Communications
Letters. Charalambous' research spans,  Stochastic
dynamical decision, estimation and control systems, information theory of stochastic systems, optimization
of stochastic systems subject to ambiguity, stochastic dynamic games and  decentralized  stochastic optimal control with  asymmetry of information, and their applications to networked control and communication systems. 
\end{IEEEbiography}


\begin{thebibliography}{10}

\bibitem{witsenhausen1971}
H.~S. Witsenhausen, ``Separation of estimation and control for discrete time
  systems,'' in {\em Proceedings of the IEEE}, pp.~1557--1566, 1971.

\bibitem{elliott-aggoun-moore1995}
R.~Elliott, L.~Aggoun, and J.~Moore, {\em Hidden Markov Models: estimation and
  Control}.
\newblock Springer, 1995.

\bibitem{bensoussan1981}
A.~Bensoussan, {\em Lecture on Stochastic Control, Lecture Notes in
  Mathematics}.
\newblock Springer-Verlag, Berlin, 1982.

\bibitem{bensoussan1992a}
A.~Bensoussan, {\em Stochastic Control of Partially Observable Systems}.
\newblock Cambridge University Press, 1992.

\bibitem{elliott-yang1991}
R.~J. Elliott and H.~Yang, ``Control of partially observed diffusions,'' {\em
  Journal on Optimization Theory and Applications}, vol.~71, no.~3,
  pp.~485--501, 1991.

\bibitem{charalambous-elliott1997}
C.~D. Charalambous and R.~Elliott, ``Certain classes of nonlinear partially
  observable stochastic optimal control problems with explicit optimal control
  laws equivalent to {LEQG/LQG} problems,'' {\em IEEE Transactions on Automatic
  Control}, vol.~42, no.~4, pp.~482--497, 1997.

\bibitem{WitsenhausenOriginal}
H.S.Witsenhausen, ``A counterexample in stochastic optimal control,'' {\em SIAM
  Journal on Control}, vol.~6, no.~1, pp.~131--147, 1968.

\bibitem{Papadimitriou-Tsitsiklis:1985}
C.~H. Papadimitriou and J.~Tsitsiklis, ``Intractable problems in control
  theory,'' {\em SIAM Journal on Control and Optimization}, vol.~24, no.~4,
  pp.~639--654, 1986.

\bibitem{hierarchialLee}
J.~T. Lee, E.~Lau, and Y.-C. Ho, ``The {W}itsenhausen counterexample: a
  hierarchical search approach for nonconvex optimization problems,'' {\em IEEE
  Transactons on Automatic Control}, vol.~46, pp.~382--397, 2001.

\bibitem{neuralnetworksolution}
M.~{Baglietto}, T.~{Parisini}, and R.~{Zoppoli}, ``Numerical solutions to the
  {W}itsenhausen counterexample by approximating networks,'' {\em IEEE
  Transactions on Automatic Control}, vol.~46, pp.~1471--1477, Sep. 2001.

\bibitem{iterativecodingWit}
J.~{Karlsson}, A.~{Gattami}, T.~J. {Oechtering}, and M.~{Skoglund}, ``Iterative
  source-channel coding approach to {W}itsenhausen's counterexample,'' in {\em
  Proceedings of the American Control Conference}, pp.~5348--5353, June 2011.

\bibitem{BasarVariations}
T.~{Basar}, ``Variations on the theme of the {W}itsenhausen counterexample,''
  in {\em Proceedings of the 47th IEEE Conference on Decision and Control},
  pp.~1614--1619, Dec 2008.

\bibitem{McEany1BadPaper}
W.~M. {McEneaney}, S.~H. {Han}, and A.~{Liu}, ``An optimization approach to the
  {W}itsenhausen counterexample,'' in {\em Proceedings of the 50th IEEE
  Conference on Decision and Control and European Control Conference},
  pp.~5023--5028, Dec 2011.

\bibitem{WuVerduTransport}
Y.~{Wu} and S.~{Verd{\'u}}, ``{W}itsenhausen's counterexample: A view from
  optimal transport theory,'' in {\em Proceedings of the 50th IEEE Conference
  on Decision and Control and European Control Conference}, pp.~5732--5737, Dec
  2011.

\bibitem{McEany2}
W.~M. McEneaney and S.~H. Han, ``Optimization formulation and monotonic
  solution method for the {W}itsenhausen problem,'' {\em Automatica}, vol.~55,
  pp.~55 -- 65, 2015.

\bibitem{Submarmanian_newresults}
V.~Subramanian, L.~Brink, N.~Jain, K.~Vodrahalli, A.~Jalan, N.~Shinde, and
  A.~Sahai, ``Some new numeric results concerning the {W}itsenhausen
  counterexample,'' in {\em Proceedings of the 56th Annual Allerton Conference
  on Communication, Control, and Computing (Allerton)}, pp.~413--420, 2018.

\bibitem{marschak1955}
J.~Marschak, ``Elements for a theory of teams,'' {\em Management Science},
  vol.~1, no.~2, 1955.

\bibitem{radner1962}
R.~Radner, ``Team decision problems,'' {\em The Annals of Mathematical
  Statistics}, vol.~33, no.~3, pp.~857--881, 1962.

\bibitem{marschak-radner1972}
J.~Marschak and R.~Radner, {\em Economic Theory of Teams}.
\newblock New Haven: Yale University Pres, 1972.

\bibitem{krainak-speyer-marcus1982a}
J.~Krainak, J.~L. Speyer, and S.~I. Marcus, ``Static team problems-part {I}:
  Sufficient conditions and the exponential cost criterion,'' {\em IEEE
  Transactions on Automatic Control}, vol.~27, no.~4, pp.~839--848, 1982.

\bibitem{krainak-speyer-marcus1982b}
J.~Krainak, J.~L. Speyer, and S.~I. Marcus, ``Static team problems-part {II}:
  Affine control laws, projections, algorithms, and the {LEGT} problem,'' {\em
  IEEE Transactions on Automatic Control}, vol.~27, no.~4, pp.~848--859, 1982.

\bibitem{kurtaran-sivan1973}
B.-Z. Kurtaran and R.~Sivan, ``Linear-{Q}uadratic-{G}aussian control with
  one-step-delay sharing pattern,'' {\em IEEE Transactions on Automatic
  Control}, vol.~19, no.~5, pp.~571--574, 1974.

\bibitem{yoshikawa1975}
T.~Yoshikawa, ``Dynamic programming approach to decentralized stochastic
  control problems,'' {\em IEEE Transactions on Automatic Control}, vol.~20,
  no.~6, pp.~796--797, 1975.

\bibitem{sandell-athans1974}
N.~R. Sandell and M.~Athans, ``Solution of some nonclassical {LQG} stochastic
  decision problems,'' {\em IEEE Transactions on Automatic Control}, vol.~19,
  no.~2, pp.~108--116, 1974.

\bibitem{varaiya-walrand1978}
P.~Varaiya and J.~Walrand, ``On delay sharing patterns,'' {\em IEEE
  Transactions on Automatic Control}, vol.~23, no.~3, pp.~443--445, 1978.

\bibitem{aicardi-davoli-minciardi1987}
M.~Aicardi, F.~Davoli, and R.~Minciardi, ``Decentralized optimal control of
  markov chains with a common past information,'' {\em IEEE Transactions on
  Automatic Control}, vol.~32, no.~11, pp.~1028--1031, 1987.

\bibitem{bansal-basar1987}
R.~Bansal and T.~Basar, ``Stochastic teams with nonclassical information
  revisited: when is an affine law optimal,'' {\em IEEE Transactions on
  Automatic Control}, vol.~32, no.~6, pp.~554--559, 1987.

\bibitem{waal-vanschuppen2000}
P.~R. de~Waal and J.~H. van Schuppen, ``A class of team problems with discrete
  action spaces: Optimality conditions based on multimodularity,'' {\em SIAM
  Journal on Control and Optimization}, vol.~38, no.~3, pp.~875--892, 2000.

\bibitem{bamieh-voulgaris2005}
B.~Bamieh and P.~Voulgaris, ``A convex characterization of distributed control
  problems in spatially invariant systems with communication constraints,''
  {\em Systems and Control Letters}, vol.~54, no.~6, pp.~575--583, 2005.

\bibitem{raoufat-djouadi:2018}
M.~E. Raoufat and S.~Djouadi, ``Optimal {H2} decentralized control of cone
  causal spatially invariant systems,'' in {\em Proceedings of the American
  Control Conference}, pp.~961--966, 2018.

\bibitem{lessard-lall2011}
L.~Lessard and S.~Lall, ``A state-space solution to the two-player optimal
  control problems,'' in {\em Proceedings of 49th Annual Allerton Conference on
  Communication, Control and Computing}, 2011.

\bibitem{vanschuppen2011}
J.~H. van Schuppen, ``Control of distributed stochastic systems-introduction,
  problems, and approaches,'' in {\em International Proceedings of the IFAC
  World Congress}, 2011.

\bibitem{vanschuppen2012}
J.~H. van Schuppen, O.~Boutin, P.~L. Kempker, J.~Komenda, T.~Masopust,
  N.~Pambakian, and A.~C.~M. Ran, ``Control of distributed systems: Tutorial
  and overview,'' {\em European Journal on Control}, pp.~579--602, 2012.

\bibitem{nayyar-mahajan-teneketzis2013}
A.~Nayyar, A.~Mahajan, and D.~Teneketzis, ``Decentralized stochastic control
  with partial history sharing: A common information approach,'' {\em IEEE
  Transactions on Automatic Control}, vol.~58, no.~7, pp.~1644--1658, 2013.

\bibitem{nedic:2017}
A.~Nedic, A.~Olshevsky, and M.~G. Rabbat, ``Network topology and
  communication-computation tradeoffs in decentralized optimization,'' {\em
  Proceedings of the IEEE}, vol.~106, no.~5, pp.~953--976, 2017.

\bibitem{ho-chu1972}
Y.-C. Ho and K.-C. Chu, ``Team decision theory and information structures in
  optimal control problems-part {I},'' {\em IEEE Transactions on Automatic
  Control}, vol.~17, no.~1, pp.~15--22, 1972.

\bibitem{ho-chu1973}
Y.-C. Ho and K.-C. Chu, ``On the equivalence of information structures in
  static and dynamic teams,'' {\em IEEE Transactions on Automatic Control},
  vol.~18, no.~2, pp.~187--188, 1973.

\bibitem{tobias2024}
O.~J. Tobias and M.~L. Treust, ``Coordination coding with causal decoder for
  vector-valued {W}itsenhausen counterexample setups,'' in {\em Proceedings of
  the 2019 IEEE Information Theory Workshop (ITW)}, August 2024.

\bibitem{witsenhausen1988}
H.~Witsenhausen, ``Equivalent stochastic control problems,'' {\em Mathematics
  of Control Signals and Systems}, vol.~1, pp.~3--11, 1988.

\bibitem{Bambos_equivalence}
C.~D. {Charalambous} and N.~U. {Ahmed}, ``Equivalence of decentralized
  stochastic dynamic decision systems via {G}irsanov's measure
  transformation,'' in {\em Proceedings of the 53rd IEEE Conference on Decision
  and Control}, pp.~439--444, 2014.

\bibitem{liptser-shiryayev1977}
R.~Liptser and A.~Shiryayev, {\em Statistics of Random Processes Vol.1}.
\newblock Springer-Verlag New York, 1977.

\bibitem{charalambous-ahmed:IEEEAC2017a}
C.~D. Charalambous and A.~N. U., ``Centralized versus decentralized
  optimization of distributed stochastic differential decision systems with
  different information structures-{P}art {I}: A general theory,'' {\em IEEE
  Transactions on Automatic Control}, vol.~62, no.~3, pp.~1194--1209, 2017.

\bibitem{charalambous-ahmed:IEEEAC2018}
C.~D. Charalambous and N.~U. Ahmed, ``Centralized versus decentralized
  optimization of distributed stochastic differential decision systems with
  different information structures—part {II}: Applications,'' {\em IEEE
  Transactions on Automatic Control}, vol.~63, no.~7, pp.~1913--1928, 2018.

\bibitem{charalambous:MCSS2016}
C.~D. Charalambous, ``Decentralized optimality conditions of stochastic
  differential decision problems via {G}irsanov's measure transformation,''
  {\em Mathematics of Control, Signals and Systems (MCSS)}, vol.~28, no.~3,
  pp.~1--55, 2016.

\bibitem{BansalWhenisAffineOptimal}
R.~{Bansal} and T.~{Basar}, ``Stochastic teams with nonclassical information
  revisited: When is an affine law optimal?,'' {\em IEEE Transactions on
  Automatic Control}, vol.~32, pp.~554--559, June 1987.

\bibitem{Grover2013}
P.~Grover, {\em Information Structures, the {W}itsenhausen Counterexample, and
  Communicating Using Actions}, pp.~1--6.
\newblock London: Springer London, 2013.

\bibitem{Lusternik}
L.~Lusternik and V.~Sobolev, ``Pr\'{e}cis d{'}analyse fonctionnelle,'' {\em
  Editions Mir}, 1989.

\bibitem{GHQGreenwood}
R.~Greenwood and J.~Miller, ``Zeros of the hermite polynomials and weights for
  gauss mechanical quadrature formula,'' {\em Bulletin of the American
  Mathematical Society}, vol.~54, no.~8, pp.~765--769, 1947.

\bibitem{MaxwellGHQ}
J.~Pimbley, ``Hermite polynomials and gauss quadrature,'' 2017.

\bibitem{GolubGHQ}
G.~H. Golub and J.~H. Welsch, ``Calculation of gauss quadrature rules,'' {\em
  Mathematics of Computation}, vol.~23, no.~106, pp.~221--s10, 1969.

\bibitem{PiecewiseCollocationAtkinson}
K.~Atkinson, I.~Graham, and I.~Sloan, ``Piecewise continuous collocation for
  integral equations,'' {\em SIAM Journal on Numerical Analysis}, vol.~20,
  no.~1, pp.~172--186, 1983.

\bibitem{elliott1982}
R.~J. Elliott, {\em Stochastic Calculus and Applications}.
\newblock Springer-Verlag, 1982.

\end{thebibliography}
\end{document}